\documentclass[a4paper,11pt]{amsart}

\usepackage{verbatim}

\usepackage[T1]{fontenc}
\usepackage[latin1]{inputenc}
\usepackage{setspace}
\usepackage{indentfirst}

\usepackage{geometry}
\usepackage[hyperfootnotes=false]{hyperref}

\usepackage{times}
\usepackage{amsthm}
\usepackage{amsmath}
\usepackage{amsfonts}
\usepackage{amssymb}
\usepackage{bibentry}
\usepackage{color}
\usepackage{mathrsfs}
\usepackage{stmaryrd}
\SetSymbolFont{stmry}{bold}{U}{stmry}{m}{n}

\usepackage{bm}

\usepackage[all]{xy}
\usepackage{graphicx}
\usepackage{subfigure}
\usepackage{todonotes}
\usepackage[toc,page]{appendix}

\newcommand{\g}{\mathfrak g}
\newcommand{\h}{\mathfrak h}

\DeclareMathOperator{\Hom}{Hom}

\DeclareMathOperator{\GL}{GL}

\DeclareMathOperator{\Aut}{Aut}

\DeclareMathOperator{\Ad}{Ad}

\DeclareMathOperator{\ad}{ad}

\newtheorem{Thm}{Theorem}[section]

\newtheorem{Pro}[Thm]{Proposition}
\newtheorem{Lem}[Thm]{Lemma}
\newtheorem{Cor}[Thm]{Corollary}
\newtheorem{Def-Pro}[Thm]{Definition-Proposition}
\newtheorem{Def}[Thm]{Definition}

\theoremstyle{definition}
\newtheorem{Ex}[Thm]{Example}
\newtheorem{Rm}[Thm]{Remark}

\begin{document}

\title{Poisson homogeneous spaces of Poisson 2-groups}

\author{Honglei Lang}
\address{College of Science, China Agricultural University, Beijing 100083, China}
\email{hllang@cau.edu.cn}
\author{Zhangju Liu}
\address{School of Mathematical and Statistical Sciences, Henan University, Kaifeng 475004, China}
\email{zhangju@henu.edu.cn}


\begin{abstract}
Drinfeld classified Poisson homogeneous spaces of a Poisson Lie group in terms of Dirac structures of the Lie bialgebra. In this paper, we study homogeneous spaces of a 2-group and develop Drinfeld theorem in the Poisson 2-group context.\end{abstract}

\maketitle
\tableofcontents

\section{Introduction}

Poisson Lie groups were introduced by Drinfeld  as the semi-classical limit of quantum groups. A Poisson Lie group is a Lie group with a Poisson structure such that the group multiplication is a Poisson map. For the Lie algebra of the Lie group, its dual space is also endowed with a Lie algebra from the Poisson structure. This pair composes a Lie bialgebra. Drinfeld proved that there is a one-one correspondence between connected, simply connected Poisson Lie groups and Lie bialgebras. 


A Poisson homogeneous space of a Poisson Lie group is a homogeneous space with a Poisson structure such that the action is a Poisson map. 
 Every Poisson homogeneous space admits an integration to a symplectic groupoid \cite{BIL}. The classification of Poisson homogeneous spaces was given by Drinfeld \cite{D2} in terms of  Dirac structures (Lagrangian subalgebras) of the corresponding Lie bialgebra. 

Drinfeld's classification theorem on homogeneous spaces  was generalized from Poisson Lie groups to Poisson groupoids  in \cite{LWX}. A Poisson groupoid is a groupoid $\Gamma$ with a Poisson structure $\Pi$ such that the graph of groupoid multiplication 
$\{(g,h,gh)|(g,h)\in \Gamma^{(2)}\}\subset \Gamma\times \Gamma\times \overline{\Gamma}$ is coisotropic, where $\overline{\Gamma}$ denotes $\Gamma$ with the opposite Poisson structure $-\Pi$. 
The infinitesimals of Poisson groupoids are Lie bialgebroids. 
In fact, there is a bijection between Poisson Lie groupoid $(\Gamma,\Pi)$ and Lie bialgebroid $(A,A^*)$, where $\Gamma\rightrightarrows M$ is an $s$-connected and simply connected Lie groupoid integrating the Lie algebroid $A$. The double $A\oplus A^*$ of a Lie bialgebroid is a Courant algebroid. In \cite{LWX}, a Dirac structure $L$ of a Lie bialgebroid $(A,A^*)$ was introduced as a subbundle of $A\oplus A^*$ which is maximally isotropic and integrable with respect to the natural pairing and Courant bracket. And it was proved that there is a bijection between Poisson homogeneous spaces of Poisson groupoids and regular Dirac structures of Lie bialgebroids. Here a homogeneous space of a Lie groupoid $\Gamma$ is of the form $\Gamma/\!\!/\Gamma_1$, where $\Gamma_1$ is a wide (containing all the identities) subgroupoid of $\Gamma$. See \cite{JL} for homogeneous spaces of Dirac groupoids. 

To describe symmetries of symmetries, the concept of (strict) Lie 2-groups was introduced, which is a Lie groupoid such that the total and base manifolds are Lie groups and the groupoid structure maps are Lie group homomorphisms; see \cite{Baez2}. Lie 2-groups can be identified with Lie group crossed modules. A Lie group crossed module $\mathbb{G}$ is a pair of Lie groups $G_0$ and $G_1$ together with a Lie group homomorphism $G_1\xrightarrow{\Phi}G_0$ and an action of $G_0$ on $G_1$ by automorphisms satisfying some compatibility conditions. 
 A Poisson 2-group \cite{CSX2} is a Lie 2-group with a Poisson structure on the total space such that  it is {simultaneously} a Poisson groupoid and a Poisson Lie group.  

Lie 2-algebras are categorification of Lie algebras \cite{Baez}. (Strict) Lie 2-algebras can be identified with Lie algebra crossed modules, which are the infinitesimal of Lie group crossed modules. Likewise, a Lie 2-bialgebra can be viewed as a pair of Lie algebra crossed modules $(\g_1\xrightarrow{\phi}\g_0)$ and $(\g^*_{0}\xrightarrow{-\phi^*}\g_1^*)$ such that $(\g_0\ltimes \g_1,\g_1^*\ltimes \g_0^*)$ is a Lie bialgebra \cite{CSX}. See also \cite{BSZ}.
In \cite{CSX2}, the authors proved that there is a one-to-one correspondence between connected and simply connected Poisson 2-groups and Lie 2-bialgebras. 

The goal of this paper is to prove the 2-categorical analogue of the Drinfeld classification theorem for Poisson homogeneous spaces. A Poisson homogeneous space of a Poisson 2-group is expected to be a homogeneous space of the Poisson Lie group and Poisson groupoid at the same time. However, note that a closed Lie 2-subgroup generally does not share the same base manifold with the original Lie 2-group. This leads to the inapplicability of the framework  in \cite{LWX} for homogeneous spaces of Poisson groupoids. This problem is tackled by considering a quotient Poisson groupoid.  For a Lie 2-group  $\mathbb{G}: G_1\to G_0$ with a closed  2-subgroup $\mathbb{H}: H_1\to H_0$, the group quotient $\mathbb{G}/\mathbb{H}$ can also be viewed as a groupoid homogeneous space of the groupoid $\mathbb{G}/H_0$. Suppose further that $\mathbb{G}$ is a Poisson 2-group and $H_0\subset \mathbb{G}$ is a Poisson subgroup so that $\mathbb{G}/H_0$ is a quotient Poisson groupoid. Then we are able to introduce Poisson homogeneous spaces of Poisson 2-groups and prove that there is a one-one correspondence between such Poisson homogeneous spaces and $H_0$-invariant Dirac structures $L$ of the Lie bialgebra $(\g_1,\g_1^*)$ such that $L\cap \g_1=\h_1$. 

On the way to our main result, we develop the theory of  Lie 2-group actions on Lie groupoids, the tangent and cotangent Lie 2-groups of a 2-group and homogeneous spaces of Lie 2-groups. These topics are of  independent interest. 

Here is the outline of this paper. 
In Section 2, we interpret a Lie 2-group action on a Lie groupoid as a homomorphism of group crossed modules. Applications include representations of Lie 2-groups on 2-vector spaces. Focusing on the adjoint and coadjoint representations, we constructed the tangent and cotangent Lie 2-groups of a 2-group; Section 3 delves into a comprehensive study of homogeneous spaces of Lie 2-groups. A homogeneous space not only has  a quotient Lie groupoid structure, but also a fiber bundle structure. Finally, in Section 4, we establish a correspondence between Poisson homogeneous spaces of Poisson 2-groups and invariant Dirac structures of some Lie bialgebras.

\section{Lie 2-group actions and representations}

In \cite{ZZ}, an action of a (strict) Lie 2-algebra on a Lie algebroid is introduced as a homomorphism of Lie algebra crossed modules. It is then integrated to  an action of a Lie 2-group on a Lie groupoid by using  Mackenzie's double Lie algebroids and groupoids. 
Later in \cite{GZ}, a Lie 2-group action is defined directly as  a group action in the category of Lie groupoids. Here we shall show the automorphism group of a Lie groupoid together with its  bisection group gives rise to a natural  group crossed module and  a Lie 2-group action can be interpreted as a  homomorphism of group crossed modules. 

In particular, for a linear Lie 2-group action, we recover the representation of a Lie 2-group on a 2-vector space. The dual of a linear Lie 2-group action is also studied.  As applications, we consider the adjoint and coadjoint actions of a Lie 2-group on its Lie 2-algebra and  the dual. The tangent and cotangent Lie 2-groups  are thus obtained.

Let us first recall the notions of Lie group (algebra) crossed modules and Lie 2-groups following \cite{Baez, Baez2}.
A {\bf Lie group crossed module} is a pair of Lie groups $G_1$ and $G_0$ with a Lie group homomorphism $\Phi: G_1\to G_0$ and an action of $G_0$ on $G_1$ by automorphisms such that
\[\Phi(g_0\triangleright g_{1})=g_0\Phi(g_1)\g_0^{-1},\qquad \Phi(g_1)\triangleright g'_1=g_1g'_1g_1^{-1},\qquad \forall g_0\in G_0,g_1,g'_1\in G_1.\]
A {\bf Lie 2-group} is a Lie groupoid $\Gamma\rightrightarrows \Gamma_0$, such that $\Gamma$ and $\Gamma_0$ are both Lie groups and all the groupoid structure maps are group homomorphisms. 

There is a bijection between Lie group crossed modules and Lie 2-groups. Given a Lie group crossed module $\mathbb{G}: G_1\to G_0$, the associated Lie 2-group is  $G_0\ltimes G_1\rightrightarrows G_0$, or simply denoted by $G_0\ltimes G_1$. The group multiplication on $G_0\ltimes G_1$ is given by 
\begin{eqnarray}\label{gpstr}
(g_0,g_1)\diamond (l_0,l_1)=(g_0l_0,(l_0^{-1}\triangleright g_1)l_1).
\end{eqnarray}
The source, target maps of the Lie groupoid  are
$s(g_0,g_1)=g_0$ and $t(g_0,g_1)=g_0\Phi(g_1)$, 
and the groupoid multiplication is 
\begin{eqnarray}\label{gpoidstr}
(g_0,g_1)*(l_0,l_1)=(g_0,g_1l_1),\qquad \mathrm{if} \quad l_0=g_0\Phi(g_1).
\end{eqnarray} 

A {\bf Lie algebra crossed module} is a pair of Lie algebras $\g_0$ and $\g_1$ with a Lie algebra homomorphism $\phi:\g_1\to \g_0$ and an action of $\g_0$ on $\g_1$ by derivations such that
\[\phi(x\triangleright u)=[x,\phi(u)],\qquad [u,v]=\phi(u)\triangleright v,\qquad \forall x\in \g_0,u,v\in \g_1.\]
Lie algebra crossed modules are equivalent to strict Lie 2-algebras \cite{Baez}. 

Throughout this paper, we use Lie 2-groups (algebras) and Lie group (algebra) crossed modules interchangeably.

\begin{Def}\cite{GZ}
Let $\mathbb{G}:G_1\xrightarrow{\Phi} G_0$ be a Lie 2-group and $P\rightrightarrows P_0$ a Lie groupoid.  A {\bf Lie 2-group action} of $\mathbb{G}$ on $P$, or a {\bf $\mathbb{G}$-action},  is a Lie group action $\triangleright: \mathbb{G}\times P\to P$  such that it is also a Lie groupoid homomorphism,
where  $\mathbb{G}\times P\rightrightarrows G_0\times P_0$ is the direct product Lie groupoid of $\mathbb{G}$ and $P$.
\end{Def}
For a $\mathbb{G}$-action, we have 
\[\iota(g_0\triangleright p_0):=\iota(g_0)\triangleright  \iota(p_0),\quad s(g\triangleright p)=s(g)\triangleright s(p),\quad t(g\triangleright p)=t(g)\triangleright t(p),\qquad \forall g\in \mathbb{G},p\in P,\]
and 
\begin{eqnarray}\label{goidmor}
(g*_{\mathbb{G}} g')\triangleright (p*p')=(g\triangleright p)*(g'\triangleright p'),\qquad \forall (g,g')\in \mathbb{G}^{(2)},(p,p')\in P^{(2)}.
\end{eqnarray}
Here $\mathbb{G}^{(2)}:=\{(g,g')\in \mathbb{G}\times \mathbb{G}; t(g)=s(g')\}$ is the set  of composable pairs.

The left action of a Lie 2-group on itself by left group translation is obviously a Lie 2-group action on a Lie groupoid,  as the group multiplication is a groupoid homomorphism.

Suppose that a  Lie 2-group $\mathbb{G}:G_1\xrightarrow{\Phi} G_0$ acts on a Lie groupoid $P\rightrightarrows P_0$ freely and properly. 
Then  the quotient space $P/ \mathbb{G}\rightrightarrows P_0/G_0$ is a Lie groupoid such that 
the projection $P\to P/\mathbb{G}$ is a Lie groupoid homomorphism.

A Lie 2-group action can also be characterized as a homomorphism of group crossed modules. 
Let $P\rightrightarrows P_0$ be a Lie groupoid. A {\it bisection} of $P$  is  a splitting $\gamma: P_0\to P$ of the source map $s$ ($s\circ \gamma=\mathrm{id}$) such that  $\phi_\gamma:=t\circ \gamma: P_0\to P_0$ is  a diffeomorphism. All bisections form a group $\mathrm{Bis}(P)$ with the multiplication and inverse given by
\begin{eqnarray}\label{bispro}
\gamma\cdot \gamma'(x)=\gamma'(x)*\gamma(\phi_{\gamma'}(x)),\qquad \gamma^{-1}(x)=\gamma(\phi_\gamma^{-1}x)^{-1},\qquad \forall x\in P_0.
\end{eqnarray}
Note that $\phi_{\gamma\cdot \gamma'}=\phi_{\gamma}\circ \phi_{\gamma'}$.
\begin{Pro}\label{newcm}
Let $P \rightrightarrows P_0$ be a Lie groupoid and $\Aut(P)$ the group of Lie groupoid automorphisms. Then we have a  group crossed module
\[\Psi: \mathrm{Bis}(P)\to \Aut(P),\qquad \Psi(\gamma)(p)=\gamma(s(p))^{-1}* p*\gamma(t(p)),\]
where the action of $\Aut(P)$ on $\mathrm{Bis}(P)$ is 
\begin{eqnarray}\label{autact}
(D\triangleright \gamma)(x)=D(\gamma D_0^{-1}(x)),\qquad D_0=D|_{P_0}\in \mathrm{Diff}(P_0).
\end{eqnarray}
\end{Pro}
\begin{proof}
First it is direct to see 
\[\big(D'\triangleright(D\triangleright \gamma)\big)(x)=D'((D\triangleright \gamma)({D'_0}^{-1}(x)))=D' D(\gamma(D_0^{-1}{D'_0}^{-1}(x))\big)=\big((D'\circ D)\triangleright \gamma\big)(x),\]
so $\triangleright$ is a group action. Then we check that $\Psi$ is a group homomorphism. For $p\in P$ with $x=s(p)$ and $y=t(p)$, we have
\begin{eqnarray*}
\Psi(\gamma\cdot \gamma')(p)&=&(\gamma\cdot \gamma')(x)^{-1}*p*(\gamma\cdot \gamma')(y)\\ &=&
\gamma(\phi_{\gamma'}(x))^{-1} *\gamma'(x)^{-1}*p*\gamma'(y)*\gamma(\phi_{\gamma'}(y))
\\ &=&\big(\Psi(\gamma)\circ \Psi(\gamma')\big)(p),
\end{eqnarray*}
which implies $\Psi(\gamma\cdot \gamma')=\Psi(\gamma)\circ \Psi(\gamma')$. It is left to show
\[\Psi(D\triangleright \gamma)=D\circ \Psi(\gamma)\circ D^{-1},\qquad \Psi(\gamma)\triangleright \gamma'=\gamma\cdot \gamma'\cdot \gamma^{-1}.\]
Actually, since $D$ is a Lie groupoid automorphism, we have
\begin{eqnarray*}
(D\circ \Psi(\gamma)\circ D^{-1})(p)&=& D\big(\gamma(D_0^{-1}x)^{-1}*D^{-1}(p)*\gamma(D_0^{-1} y)\big)\\ &=&D(\gamma(D_0^{-1}x))^{-1}*p*D(\gamma(D_0^{-1} y))\\ &=&
(D\triangleright \gamma)(x)^{-1}*p*(D\triangleright \gamma)(y)\\ &=&\Psi(D\triangleright \gamma)(p).
\end{eqnarray*}
Also, observing that $\Psi(\gamma)(x)=\phi_{\gamma}(x)$, we have
\begin{eqnarray*}
\big(\Psi(\gamma)\triangleright \gamma'\big)(x)&=&\Psi(\gamma)(\gamma' \phi_{\gamma}^{-1}(x))\\ &=&\gamma(\phi_{\gamma}^{-1}x)^{-1}*\gamma' (\phi_{\gamma}^{-1}(x))*\gamma(\phi_{\gamma'}\phi_{\gamma}^{-1}x)\\ &=&\gamma^{-1}(x)*(\gamma\cdot \gamma' )(\phi_{\gamma}^{-1}(x))
\\ &=&(\gamma\cdot \gamma'\cdot \gamma^{-1})(x).
\end{eqnarray*}
\end{proof}

\begin{Thm}\label{cmhomo}
Given a Lie 2-group  $\mathbb{G}$ and a Lie groupoid  $P$, there is a one-one correspondence between Lie 2-group actions $\triangleright$ of $\mathbb{G}$ on $P$ and homomorphisms $(F_0,F_1)$ of  group crossed modules:
\begin{equation*}
		\vcenter{\xymatrix{
			G_1 \ar[d]^{\Phi} \ar[r]^{F_1} &\mathrm{Bis}(P) \ar[d]^{\Psi}\\
			G_0\ar[r]^{F_0} &\Aut(P)
		}}.\end{equation*}
They are related by 
\begin{eqnarray}\label{F01}
F_0(g_0)(p)=g_0\triangleright p,\qquad 	F_1(g_1)(x)=g_1\triangleright x, 
	\qquad \forall x\in P_0,g_0\in G_0, g_1\in G_1,p\in P.
	\end{eqnarray}
\end{Thm}
\begin{proof}
Given a Lie 2-group action $\triangleright$, we show that $(F_0,F_1)$ defined in \eqref{F01} is a crossed module homomorphism. First,  since the $\mathbb{G}$-action  is a Lie group action, we have
\[F_0(g_0 g'_0)(p)=(g_0g'_0)\triangleright p=g_0\triangleright (g'_0\triangleright p)=\big(F_0(g_0)\circ F_0(g'_0)\big) (p).\]
Also using the fact that $\phi_{F_1(g_1)}(x)=t(g_1\triangleright x)=\Phi(g_1) \triangleright x$, we have
\begin{eqnarray*}
\big(F_1(g_1)\cdot F_1(g'_1)\big)(x)&=&F_1(g'_1)(x)*F_1(g_1)(\Phi(g'_1)\triangleright x)\\ &=&
(g'_1\triangleright x)* (({g_1}\diamond \Phi(g'_1))\triangleright x)\\ &=&\big((e,g'_1)*_{\mathbb{G}}(\Phi(g'_1), \Phi({g'_1}^{-1})\triangleright {g_1})\big)\triangleright x\\ &=&(g_1g'_1)\triangleright x=F_1(g_1g'_1)(x),
\end{eqnarray*}
where we used \eqref{goidmor} in the third equation. So $F_0$ and $F_1$ are group homomorphisms. Then we show that the diagram is commutative, i.e. $F_0\circ \Phi=\Psi\circ F_1$. Indeed, by using \eqref{goidmor}, we have 
\begin{eqnarray*}
\Psi(F_1(g_1))(p)&=&(F_1(g_1)(x))^{-1}*p*F_1(g_1)(y)\qquad x=s(p), y=t(p)\\ &=&(g_1\triangleright x)^{-1}*p*(g_1\triangleright y)\\ &=&
((\Phi(g_1),g_1^{-1})\triangleright x)*p*(g_1\triangleright y)
\\ &=&(\Phi(g_1),e)\triangleright p=F_0(\Phi(g_1))(p).
\end{eqnarray*}
It is left to check
$F_1(g_0\triangleright g_1)=F_0(g_0)\triangleright F_1(g_1)$. 
Indeed, we have
\begin{eqnarray*}
\big(F_0(g_0)\triangleright F_1(g_1)\big)(x)&=&F_0(g_0)(F_1(g_1)(g_0^{-1}\triangleright x))=
F_0(g_0)((g_1\diamond g_0^{-1})\triangleright x)\\ &=&(g_0 \diamond g_1\diamond g_0^{-1})\triangleright x=(g_0\triangleright g_1)\triangleright x=F_1(g_0\triangleright g_1)(x).
\end{eqnarray*}
This shows that $(F_0,F_1)$ is a homomorphism of crossed modules. 

Conversely,  given a crossed module homomorphism $(F_0,F_1)$, define 
\begin{eqnarray}\label{actF}
\triangleright: \mathbb{G}\times P\to P,\qquad 
(g_0,g_1)\triangleright p:=F_0(g_0)\big(p*F_1(g_1)(y)\big),\quad y=t(p).
\end{eqnarray}
We claim  that $\triangleright$ is a Lie 2-group action. Namely, it is a Lie group action and also  a Lie groupoid homomorphism (only check \eqref{goidmor}). 
For $(g_0,g_1),(l_0,l_1)\in G_0\times G_1$, since $F_0,F_1$ are group homomorphisms and 
\[F_1(l_0^{-1}\triangleright g_1)=F_0(l_0^{-1})\triangleright F_1(g_1),\qquad \phi_{F_1(l_1)}(x)=\Psi(F_1(l_1))x=F_0(\Phi(l_1))x,\qquad \forall x\in P_0,\]
we have 
\begin{eqnarray*}
\big((g_0,g_1)\diamond (l_0,l_1)\big)\triangleright p
&=&
F_0(g_0l_0)\big(p*F_1((l_0^{-1}\triangleright g_1)l_1)(y)\big)\\ &=&F_0(g_0l_0)\big(p*(F_0(l_0^{-1})\triangleright F_1(g_1))\cdot F_1(l_1)(y)\big)\\ &=&F_0(g_0l_0)\big(p*F_1(l_1)(y)*(F_0(l_0^{-1})\triangleright F_1(g_1))(F_0(\Phi(l_1))y)\\ &=&F_0(g_0l_0)\big(p*F_1(l_1)(y)*F_0(l_0^{-1})(F_1(g_1)(F_0(l_0\Phi(l_1))y)\big)\\ &=&
F_0(g_0)\big((F_0(l_0)(p*F_1(l_1)y))*(F_1(g_1)F_0(l_0)F_0(\Phi (l_1))y)\big)\\ &=&(g_0,g_1)\triangleright \big(F_0(l_0) (p*F_1(l_1)y)\big)=(g_0,g_1)\triangleright \big((l_0,l_1)\triangleright p\big),\end{eqnarray*}
where we have also used \eqref{bispro},\eqref{autact} and the fact that $F_0(g_0)\in \Aut(P)$. Thus $\triangleright$ is a group action. Then,  for a composable pair $(p,p')\in P^{(2)}$, letting $y=t(p)=s(p'), z=t(p')$, we find
\begin{eqnarray*}
&&\big((g_0,g_1)\triangleright p\big)*\big((g_0\Phi(g_1),l_1)\triangleright p'\big)\\ &=&
F_0(g_0)\big(p*F_1(g_1)(y)\big)*F_0(g_0\Phi(g_1))\big(p'*F_1(l_1)z\big)\\ &=&
F_0(g_0)\big(p*F_1(g_1)(y)*\Psi(F_1(g_1))(p')*\Psi(F_1(g_1))(F_1(l_1)z)\big)\\ &=&
F_0(g_0)\big(p*p'* F_1(g_1)z* (F_1(g_1)z)^{-1}*F_1(l_1)z*F_1(g_1)(\phi_{F_1(l_1)}z)\big)\\ &=&
F_0(g_0)\big(p*p'* (F_1(g_1)\cdot F_1(l_1))z\big)
\\ &=&(g_0,g_1l_1)\triangleright (p*p')=\big((g_0,g_1)*_\mathbb{G} (g_0\Phi(g_1),l_1)\big)\triangleright (p*p')
,\end{eqnarray*}
which implies \eqref{goidmor}. Hence we obtain a Lie 2-group action from a crossed module homomorphism.

It suffices to show the above two processes are inverse to each other. Given an action $\triangleright$, we define a crossed module homomorphism $(F_0,F_1)$ by \eqref{F01}, which then gives rise to  a Lie 2-group action $\triangleright^F$ by \eqref{actF}. We claim $\triangleright^F=\triangleright$. In fact, we have 
\begin{eqnarray*}
(g_0,g_1)\triangleright^F p&=&F_0(g_0)\big(p*F_1(g_1)(y)\big)=g_0\triangleright \big(p*(g_1\triangleright y)\big)\\ &=&
g_0\triangleright \big((e,e)*_{\mathbb{G}}(e,g_1))\triangleright (p*y)\big)\\ &=&g_0\triangleright (g_1\triangleright p)=
(g_0,g_1)\triangleright p,
\end{eqnarray*}
where we have used the fact that $\triangleright$ is a Lie group action and also a Lie groupoid homomorphism. This completes the proof.
\end{proof}

Denote by $A$ the Lie algebroid of the Lie groupoid $P$.  
The  infinitesimal of the   group crossed module $\mathrm{Bis}(P)\xrightarrow{\Psi} \Aut(P)$ in Proposition \ref{newcm} is 
the Lie algebra crossed module
\[\psi: \Gamma(A)\to \mathrm{Der}(A),\qquad \psi(u)=[u,\cdot]_A,\]
where  $\mathrm{Der}(A)$ acts on $\Gamma(A)$ naturally.

 Given a $\mathbb{G}$-action on $P$,  
taking the differential of the homomorphism in Theorem \ref{cmhomo}, we obtain a homomorphism of Lie algebra crossed modules 
\begin{equation*}
		\vcenter{\xymatrix{
			\g_1 \ar[d]^{\phi} \ar[r]^{ } &\Gamma(A) \ar[d]^{\psi}\\
			\g_0\ar[r]^{ } &\mathrm{Der}(A)
		}}.
		\end{equation*}
This gives an action of a Lie 2-algebra $\g$ on a Lie algebroid $A$ as introduced in \cite{ZZ}.

\begin{Ex}
If the Lie groupoid $P$ is the pair groupoid $M\times M\rightrightarrows M$ for a manifold $M$,  the crossed module $\mathrm{Bis}(P)\xrightarrow{\Psi} \Aut(P)$ becomes $\mathrm{Diff}(M)\xrightarrow{\mathrm{Id}} \mathrm{Diff}(M)$, where  $\mathrm{Diff}(M)$ adjoint acts on itself.
A crossed module homomorphism from $\mathbb{G}$ to $\mathrm{Diff}(M)\xrightarrow{\mathrm{Id}} \mathrm{Diff}(M)$ can only be of the form
\begin{equation*}
		\vcenter{\xymatrix{
			G_1 \ar[d]^{\Phi} \ar[r]^{\sigma\circ \Phi} &\mathrm{Diff}(M) \ar[d]^{\mathrm{Id}}\\
			G_0\ar[r]^{\sigma} &\mathrm{Diff}(M)
		}},
		\end{equation*}
where $\sigma: G_0\to \mathrm{Diff}(M)$ is a Lie group action of $G_0$ on $M$. Therefore,  by \eqref{actF}, a Lie 2-group action of $\mathbb{G}$ on the pair groupoid $M\times M\rightrightarrows M$  is determined by such a $\sigma$ in the way:  \[(G_0\ltimes G_1)\times (M\times M)\to M\times M,\qquad (g_0,g_1)\triangleright(m_1,m_2)=\big(\sigma(g_0)m_1,\sigma(g_0\Phi(g_1))m_2\big).\]
\end{Ex}
	\begin{Ex}
When the Lie groupoid $P$ is a Lie group $H$, the group crossed module $\mathrm{Bis}(P)\xrightarrow{\Psi} \Aut(P)$ turns out to be the well-known Lie group crossed module $H\xrightarrow{\mathrm{AD}} \Aut(H)$ (after a minor adjustment).	
A Lie 2-group action of $\mathbb{G}$ on a Lie group $H$ is determined by a pair $(F_0,F_1)$:
\begin{equation*}
		\vcenter{\xymatrix{
			G_1 \ar[d]^{\Phi} \ar[r]^{F_1} &H \ar[d]^{\mathrm{AD}}\\
			G_0\ar[r]^{F_0} &\Aut(H)
		}},
		\end{equation*}
where $F_0$ and $F_1$ are Lie group homomorphisms satisfying that 
	\[F_0(\Phi(g_1))(h)=F_1(g_1)hF_1(g_1)^{-1},\qquad F_1(g_0\triangleright g_1)=F_0(g_0)(F_1(g_1)).\]The Lie 2-group action is by \eqref{actF} given by
	\[\triangleright: (G_0\ltimes G_1)\times H\to H,\qquad (g_0,g_1)\triangleright h=F_0(g_0)\big(hF_1(g_1)\big).\]
	\end{Ex}
\begin{Ex}\label{cmcm} 
Let the Lie groupoid $P$ be a Lie 2-group $\mathbb{H}: H_0\ltimes H_1\rightrightarrows H_0$. Consider bisections of $\mathbb{H}$ that are also group homomorphisms and groupoid automorphisms of $\mathbb{H}$ that are also crossed module automorphisms, denoted by $D(H_0,H_1)$ and $\Aut(\mathbb{H})$, respectively. Then we obtain a subcrossed module $D(H_0,H_1)\to \Aut(\mathbb{H})$ of $\mathrm{Bis}(P)\xrightarrow{\Psi} \Aut(P)$. This subcrossed module is called the actor crossed module of $\mathbb{H}$ and was introduced to define actions of a Lie 2-group on another Lie 2-group in \cite{N}.
\end{Ex}
Let us focus on the particular case of Example \ref{cmcm} when $\mathbb{H}$ is a 2-vector space $\mathbb{V}$. We shall see that linear Lie 2-group actions are exactly  representations of this Lie 2-group on  2-vector spaces. 

A 2-vector space $\mathbb{V}: V_1\xrightarrow{\phi} V_0$ gives an action Lie groupoid $\mathbb{V}: V_0\times V_1\rightrightarrows V_0$, whose source, target and multiplication are 
\begin{eqnarray}\label{linear2gp}
s(x,u)=x,\qquad t(x,u)=x+\phi(u),\qquad (x,u)*(x+\phi(u),v)=(x,u+v)
\end{eqnarray}
for $x\in V_0,u,v\in V_1$. With  the abelian group structure, it is a linear Lie 2-group.  A $\mathbb{G}$-action on $\mathbb{V}$ is {\bf linear} if $g\triangleright (\cdot): \mathbb{V}\to \mathbb{V}$ is a linear map for all $g\in \mathbb{G}$.

\begin{Pro}\label{sd2gp}Given a linear action $\triangleright$ of a Lie 2-group $\mathbb{G}$ on $\mathbb{V}$, then \begin{itemize}
\item[\rm (i)]  there is  a semi-direct product Lie 2-group $\mathbb{G}\ltimes \mathbb{V}\rightrightarrows G_0\ltimes V_0$, with the direct product groupoid structure and  the semi-direct product group structure defined by 
\begin{eqnarray*}
(g,\omega)*(l,\delta)&=&(g*_{\mathbb{G}} l, \omega*\delta),\qquad (g,l)\in \mathbb{G}^{(2)}, (\omega,\delta)\in \mathbb{V}^{(2)},\\
(g,\omega)\diamond (g',\omega')&=&(g\diamond g',\omega+g\triangleright \omega'),\qquad \forall g,g'\in \mathbb{G},\omega,\omega'\in \mathbb{V}.
\end{eqnarray*}
\item[\rm (ii)]
the Lie group crossed module associated to $\mathbb{G}\ltimes \mathbb{V}$ is 
$G_1\ltimes V_1\xrightarrow{\Phi\times \phi} G_0\ltimes V_0$, 
where the group structures and the action are defined by 
\begin{eqnarray*}
(g_1,u)\cdot (g'_1,u')&=&(g_1g'_1,u+g_1\triangleright u'),\\ 
(g_0,x)\cdot (g'_0,x')&=&(g_0g'_0, x+g_0 \triangleright x'),\\ 
(g_0,x)\triangleright  (g_1,u)&=&(g_0\triangleright g_1, g_0 \triangleright u-(g_0\triangleright g_1)\triangleright x+x),
\end{eqnarray*}
for $g_0,g'_0\in G_0,g_1,g'_1\in G_1,x,x'\in V_0$ and $u,u'\in V_1$.
\end{itemize}
\end{Pro}
\begin{proof}
For (i), we only verify that the groupoid multiplication is a group homomorphism, i.e.
\begin{eqnarray}\label{sd}
\big((g,\omega)\diamond (g',\omega')\big)*\big((l,\delta)\diamond (l',\delta')\big)=\big((g,\omega)*(l,\delta)\big)\diamond \big((g',\omega')*(l',\delta')\big),
\end{eqnarray}
for composable pairs $(g,l), (g',l')\in \mathbb{G}^{(2)}$ and $(\omega,\delta),(\omega',\delta')\in \mathbb{V}^{(2)}$. By \eqref{goidmor} and the fact that $\mathbb{G}, \mathbb{V}$ are Lie 2-groups,  the left hand side of \eqref{sd} equals to 
\begin{eqnarray*}
(gg', \omega+g\triangleright \omega')*(ll',\delta+l\triangleright \delta')&=&\big((gg')*_{\mathbb{G}}(ll'),(\omega+g\triangleright \omega')*(\delta+l\triangleright \delta')\big)\\ &=&\big((g*_{\mathbb{G}}l)(g'*_{\mathbb{G}}l'), \omega*\delta+(g\triangleright \omega')*(l\triangleright \delta')\big)\\ &=&\big((g*_{\mathbb{G}}l)(g'*_{\mathbb{G}}l'), \omega*\delta+(g*_{\mathbb{G}}l)\triangleright (\omega'*\delta')\big)\\ &=&\big((g,\omega)*(l,\delta)\big)\diamond \big((g',\omega')*(l',\delta')\big),
\end{eqnarray*} 
which matches with the right hand side of  \eqref{sd}. 

For (ii),  from a Lie 2-group $\Gamma \rightrightarrows \Gamma_0$,  a crossed module $\ker s\xrightarrow{t} \Gamma_0$ can be constructed with the group action of $\Gamma_0$ on $\ker s$ defined by $g_0\triangleright g_1:=g_0\diamond g_1\diamond g_0^{-1}$, where $\diamond$ is the group multiplication in $\Gamma$. 

Applying this construction to our Lie 2-group $\mathbb{G}\ltimes \mathbb{V}$, we have $\ker s=G_1\ltimes V_1$ and $t|_{\ker s}=\Phi\times \phi$. The group multiplications on $G_i\ltimes V_i$ for $i=0,1$ is direct to obtain. 
For the action of $G_0\ltimes V_0$ on $G_1\ltimes V_1$, we have
\begin{eqnarray*}
(g_0,x)\triangleright  (g_1,u)&:=&(g_0,x)\diamond (g_1,u)\diamond (g_0^{-1}, -g_0^{-1}\triangleright x)=(g_0\diamond g_1,x+g_0\triangleright u)\diamond (g_0^{-1}, -g_0^{-1}\triangleright x)\\ &=&\big(g_0\triangleright  g_1, x+g_0\triangleright u-(g_0\triangleright  g_1)\triangleright x\big)\in G_1\times V_1,
\end{eqnarray*}
as desired. The rest is obvious.\end{proof}

Associated to a 2-vector space $\mathbb{V}: V_1\xrightarrow{\phi} V_0$, we have the general linear Lie 2-group  $\GL(\mathbb{V}): \GL_1(\mathbb{V})\xrightarrow{\partial} \GL_0(\mathbb{V})$, where 
\begin{eqnarray*}
\GL_1(\mathbb{V})&:=&\{\gamma\in \Hom(V_0,V_1); \mathrm{Id}+\phi\circ \gamma\in \GL(V_0)\},\\ 
\GL_0(\mathbb{V})&:=&\{(A_0,A_1)\in \GL(V_0)\oplus \GL(V_1); A_0\circ \phi=\phi\circ A_1\},\\
\partial(\gamma)&:=&\big(\mathrm{Id}+\phi\circ \gamma, \mathrm{Id}+\gamma\circ \phi\big).
\end{eqnarray*}
Here the group structures on $\GL_i(\mathbb{V})$ are \[\gamma\cdot \gamma'=\gamma+\gamma'+\gamma\circ \phi\circ \gamma',\qquad (A_0,A_1)\circ (A'_0,A'_1)=(A_0\circ A'_0,A_1\circ A'_1),\]
 and the action of $\GL_{0}(\mathbb{V})$ of $\GL_1(\mathbb{V})$ is given by 
\[ (A_0,A_1)\triangleright \gamma=A_1\circ \gamma\circ A_0^{-1}.\]
The Lie 2-algebra of $\GL(\mathbb{V})$ is $\mathrm{gl}(\mathbb{V})$. See \cite{SZ} for details.
A {\bf representation} of a Lie 2-group $\mathbb{G}$ on $\mathbb{V}$ is a crossed module homomorphism $\sigma: \mathbb{G}\to \mathbb{GL}(\mathbb{V})$.

\begin{Pro}\label{linearcm}
There is a one-one correspondence between  linear Lie 2-group actions of $\mathbb{G}$ on $\mathbb{V}$ and representations of  $\mathbb{G}$ on the 2-vector space $\mathbb{V}$.
\end{Pro}
\begin{proof}
We first show that $\GL(\mathbb{V})$ is a subcrossed module of  $\mathrm{Bis}(\mathbb{V})\xrightarrow{\Psi} \Aut(\mathbb{V})$. It is obvious that  $\GL_0(\mathbb{V})$ is the set of linear groupoid automorphisms, which is a subgroup of $\Aut(\mathbb{V})$. Then by identifying $\gamma\in \GL_1(\mathbb{V})$ with the bisection $x\mapsto (x,\gamma(x))$, we see that $\GL_1(\mathbb{V})$ can be viewed as the set of linear bisections of $\mathbb{V}$. For two  linear bisections $\gamma$ and $\gamma'$, by Proposition \ref{newcm}, we have $\phi_{\gamma}=\mathrm{Id}+\phi\circ \gamma$ and 
\[\gamma\cdot \gamma'(x)=(x,\gamma'(x))*(\phi_{\gamma'}(x),\gamma(\phi_{\gamma'}(x))=(x,\gamma(x)+\gamma'(x)+(\gamma\circ \phi\circ \gamma')(x)),\]
which implies that $\gamma\cdot \gamma'=\gamma+\gamma'+\gamma\circ \phi\circ \gamma'$. Hence one obtains that $\GL_1(\mathbb{V})\subset \mathrm{Bis}(\mathbb{V})$ is a Lie subgroup.  Then by the fact that $(x,u)^{-1}=(x+\phi(u),-u)$, we compute
\begin{eqnarray*}
\Psi(\gamma)(x,u) &=&
(x+\phi\gamma(x),-\gamma(x))
*(x,u)*(x+\phi(u),\gamma(x+\phi(u))\\ &=&
((\mathrm{Id}+\phi\circ \gamma)x,(\mathrm{Id}+\gamma\circ \phi)u)=\partial(\gamma)(x,u),\end{eqnarray*}
which shows $\Psi=\partial$.
Then we have \[((A_0,A_1)\triangleright \gamma)(x)=((A_0,A_1)\circ\gamma\circ A_0^{-1})x=(x,A_1\gamma A_0^{-1}x).\]
Hence we see $\GL(\mathbb{V})$ is a subcrossed module. 

If  the Lie 2-group action is linear, the image of $(F_0,F_1)$ from $\mathbb{G}$ to $\mathrm{Bis}(\mathbb{V})\xrightarrow{\Psi} \Aut(\mathbb{V})$ in Theorem \ref{cmhomo} falls into the subcrossed module $\mathrm{GL}(\mathbb{V})$, which gives a representation. It is moreover a one-one correspondence.
\end{proof}

Explicitly, given a linear Lie 2-group action $\triangleright: \mathbb{G}\times \mathbb{V}\to \mathbb{V}$, by \eqref{F01}, a representation $(\sigma^0, \sigma^1): \mathbb{G}\to \GL(\mathbb{V})$ is given by
\begin{eqnarray}\label{sigma} 
\sigma^0_{g_0}(x,u):=(g_0\triangleright x,g_0\triangleright u),\qquad \sigma^1_{g_1}(x):=\mathrm{pr}_{V_1}(g_1\triangleright x)=g_1\triangleright x-x,
\end{eqnarray}
for all $x\in V_0,u\in V_1$. Conversely, by \eqref{actF}, for a representation $(\sigma^0, \sigma^1): \mathbb{G}\to \GL(\mathbb{V})$, a linear $\mathbb{G}$-action $\triangleright: \mathbb{G}\times \mathbb{V}\to \mathbb{V}$ is given by
\[(g_0,g_1)\triangleright (x,u)=\big(\sigma^0_{g_0}(x),\sigma^0_{g_0} (u+\sigma^1_{g_1}(x+\phi u))\big)=\big(\sigma^0_{g_0} x,\sigma^0_{g_0} \sigma^1_{g_1}x+\sigma^0_{g_0\Phi(g_1)} u\big).\]

 For instance, 
the Lie 2-group $\GL(\mathbb{V})$ represents naturally on $\mathbb{V}$ with  $(\mathrm{Id},\mathrm{Id}): \GL(\mathbb{V})\to \GL(\mathbb{V})$. The associated linear Lie 2-group action $\triangleright: \GL(\mathbb{V})\times \mathbb{V}\to \mathbb{V}$ is 
\[\qquad (A,\gamma)\triangleright (x,u)=\big(A_0x,A_1(u+\gamma x+ \gamma \phi u)\big),\] for all $A_0\in \mathrm{GL}(V_0), A_1\in \mathrm{GL}(V_1)$ and $\gamma\in \Hom(V_0,V_1)$.

We remark that the crossed module in Proposition \ref{sd2gp} is the semi-direct product crossed module induced from the representation \eqref{sigma}; see \cite{N}.

One approach to obtain a dual Lie 2-group action of  $\mathbb{G}$ on $\mathbb{V}^*$ is by taking the dual of a representation and then applying Theorem \ref{cmhomo}.

\begin{Pro}\label{dual}
Let  $\sigma=(\sigma^0,\sigma^1):\mathbb{G}\to \GL(\mathbb{V})$ be a representation of a Lie 2-group $\mathbb{G}$ on a 2-vector space $\mathbb{V}$.
\begin{itemize}
\item[\rm (i)]Taking the dual, we obtain a representation $\sigma^*$ of $\mathbb{G}$ on $\mathbb{V^*}$:\begin{equation*}
		\vcenter{\xymatrix{
			G_1 \ar[d]^{\Phi} \ar[r]^{\sigma^{*1}} &\GL_1(\mathbb{V}^*)\ar[d]^{\partial}\\
			G_0\ar[r]^{\sigma^{*0}} &\GL_0(\mathbb{V}^*)		}},
	\end{equation*}
where 
\[\sigma^{*0}_{g_0}:=(\sigma^0_{g_0^{-1}})^*,\qquad \sigma^{*1}_{g_1}:=(\sigma^1_{g_1^{-1}})^*.\]
This representation is called the {\bf dual representation} of  $\sigma$.
\item[\rm (ii)] We have a dual Lie 2-group action $\sigma^{2gp*}:\mathbb{G}\times \mathbb{V^*}\to \mathbb{V^*}$, which is given by 
\begin{eqnarray}\label{2gp dual}
\sigma^{2gp*}_{(g_0,g_1)}(\xi,\alpha)
=(\sigma_{g_0}^{*0}\xi, \sigma_{g_0}^{*0}\sigma_{g_1}^{*1}\xi+\sigma_{g_0\Phi(g_1)}^{*0}\alpha), \qquad \forall \xi\in \g_1^*, \alpha\in \g_0^*.
\end{eqnarray}
\end{itemize}\end{Pro}
\begin{proof}
(i) follows from a straightforward verification. (ii) follows from Theorem \ref{cmhomo} and Formula \eqref{actF}.\end{proof}

\begin{Rm}\label{sn} 
In order to obtain a dual Lie 2-group action, we should take the duals of  functors ($\mathrm{GL}_0(\mathbb{V})$) and natural transformations ($\mathrm{GL}_1(\mathbb{V})$) as in Proposition \ref{dual}, instead of taking the dual of linear maps ($\mathrm{GL}(\mathbb{V})$) on the whole vector space $\mathbb{V}$. 

In fact, considering  the dual  of the group action, we obtain  $\triangleright^*: \mathbb{G}\times \mathbb{V}^*\to \mathbb{V}^*$, which is given by 
\[\langle g\triangleright^* \gamma,\omega\rangle=\langle \gamma, g^{-1}\triangleright \omega\rangle, \qquad \forall g\in \mathbb{G}, \gamma\in \mathbb{V}^*, \omega\in \mathbb{V}.\]
This is in general not a Lie 2-group action. Actually, $\triangleright^*$ is a Lie 2-group action if and only if 
\begin{eqnarray}\label{iff}
 g_1\triangleright u=u,\qquad \Phi(g_1)\triangleright x=x,\qquad \forall g_1\in G_1, x\in V_0, u\in V_1.
 \end{eqnarray}
 \end{Rm}

For a Lie 2-group $\mathbb{G}$ with its Lie 2-algebra $\g$, denote by $\Ad$ and $\Ad^*$ the adjoint and coadjoint actions of the Lie group $\mathbb{G}$ on $\g$ and $\g^*$, respectively. 
\begin{Pro}\label{ad2action}
\begin{itemize}
\item [\rm (i)] For any $(g_0,g_1)\in G_0\ltimes G_1$ and $(x,u)\in \g_0\oplus \g_1$, we have
\begin{eqnarray}
\label{adjoint}\Ad_{(g_0,g_1)}(x,u)&=&\big(\Ad^{G_0}_{g_0} x, g_0\triangleright (\Ad_{g_1}^{G_1} u-R_{g_1^{-1}*}\widehat{x}_{g_1})\big),
\end{eqnarray}
where  $\widehat{x}_{g_1}\in T_{g_1} G_1$ is defined by $\widehat{x}_{g_1}
:=\frac{d}{dt}|_{t=0} (\exp^{tx}\triangleright g_1)$.
\item [\rm (ii)] The adjoint action $\Ad: \mathbb{G}\times \g\to \g$ is a Lie 2-group action.
\end{itemize}\end{Pro}
\begin{proof}
(i) holds by the calculation
\begin{eqnarray*}
\Ad_{(g_0,g_1)}(x,u)&=&\frac{d}{dt}|_{t=0} (g_0,g_1)\diamond (\exp^{tx},\exp^{tu})\diamond (g_0^{-1},g_0\triangleright g_1^{-1})
\\ &=&\frac{d}{dt}|_{t=0}\big(g_0\exp^{tx}g_0^{-1},g_0\triangleright ((\exp^{-tx}\triangleright g_1)\exp^{tu}g_1^{-1})\big)
\\ &=&\big(\Ad^{G_0}_{g_0} x, g_0\triangleright (\Ad_{g_1}^{G_1} u-R_{g_1^{-1}*}\widehat{x}_{g_1})\big).
\end{eqnarray*} For (ii), we show the adjoint action 
$\Ad$
is a Lie groupoid homomorphism. We only check
\begin{eqnarray}\label{adgoid}
\big(\Ad_{(g_0,g_1)} (x,u) \big)*\big(\Ad_{(g_0\Phi(g_1),h_1)} (x+\phi(u),v) \big)=\Ad_{(g_0,g_1h_1)} (x,u+v),
\end{eqnarray}
for $g_0\in G_0,g_1,h_1\in G_1$ and $x\in \g_0$ and $u,v\in \g_1$. 
By \eqref{adjoint}, 
 the left hand side of \eqref{adgoid} equals 
\begin{eqnarray*}
&&\big(\Ad^{G_0}_{g_0} x,g_0\triangleright (\Ad^{G_1}_{g_1}u-R_{g_1^{-1}*}\widehat{x}_{g_1})\big)*\\ && \big(\Ad^{G_0}_{g_0\Phi(g_1)} (x+\phi(u)),g_0\Phi(g_1)\triangleright (\Ad^{G_1}_{ h_1}v-R_{h_1^{-1}*}\widehat{(x+\phi(u))}_{h_1})\big)
\\ &=&\big(\Ad^{G_0}_{g_0} x,g_0\triangleright (\Ad^{G_1}_{g_1}u-R_{g_1^{-1}*}\widehat{x}_{g_1})+g_0\Phi(g_1)\triangleright (\Ad^{G_1}_{ h_1}v-R_{h_1^{-1}*}\widehat{(x+\phi(u))}_{h_1})\big),
\end{eqnarray*} 
and the right hand side of \eqref{adgoid} is 
\[\big(\Ad^{G_0}_{g_0} x,g_0\triangleright (\Ad^{G_1}_{g_1h_1}(u+v)-R_{(g_1h_1)^{-1}*}\widehat{x}_{g_1h_1}\big).\]
Using the fact that $\Phi(g_1)\triangleright h_1=g_1h_1g_1^{-1}$ and $g_0\triangleright (g_1h_1)=(g_0\triangleright g_1)(g_0\triangleright h_1)$, we have
\begin{eqnarray*}
g_0\Phi(g_1)\triangleright \Ad^{G_1}_{h_1}v=g_0\triangleright (\Ad^{G_1}_{g_1h_1}v),
\end{eqnarray*}
and 
\begin{eqnarray*}
g_0\Phi(g_1)\triangleright (R_{ h_1^{-1}*}\widehat{\phi(u)}_{h_1})
&=&\frac{d}{dt}|_{t=0} g_0\triangleright \big(g_1 (\Phi(\exp^{tu})\triangleright h_1) h_1^{-1}g_1^{-1}\big)\\ &=&\frac{d}{dt}|_{t=0} g_0\triangleright (g_1 \exp^{tu}h_1 \exp^{-tu} h_1^{-1}g_1^{-1})\\ &=&g_0\triangleright (\Ad^{G_1}_{g_1} u-\Ad^{G_1}_{g_1h_1} u),
\end{eqnarray*}
and also
\begin{eqnarray*}
g_0\triangleright (R_{(g_1h_1)^{-1}*}\widehat{x}_{g_1h_1})&=&\frac{d}{dt}|_{t=0}g_0\triangleright\big((\exp^{tx}\triangleright g_1)(\exp^{tx}\triangleright h_1)h_1^{-1}g_1^{-1}\big)
\\ &=&g_0\triangleright  (R_{g_1^{-1}*} \widehat{x}_{g_1})+g_0\Phi(g_1)\triangleright (R_{h_1^{-1}*}\widehat{x}_{h_1}).\end{eqnarray*}
These three equations show that the terms involving $v,u$ and $x$ of both sides of \eqref{adgoid} are equal respectively. Hence \eqref{adgoid} holds by linearity.
\end{proof}

By Propositions \ref{sd2gp}, associated to the adjoint action $\Ad$, we have a semi-direct product Lie 2-group $\mathbb{G}\ltimes_{\Ad} \g$. We call it  the {\bf tangent 2-group} of a Lie 2-group $\mathbb{G}$. 
On the other hand, 
the tangent bundle $T\mathbb{G}\rightrightarrows TG_0$ of a Lie 2-group $\mathbb{G}$ is naturally a Lie 2-group with the tangent Lie group and Lie groupoid structures. In fact, we have an isomorphism $T\mathbb{G}\cong \mathbb{G}\ltimes \g$ of Lie 2-groups.

Applying Proposition \ref{linearcm} to the adjoint action of $\mathbb{G}$ on its Lie 2-algebra $\g$, we have 
\begin{Cor}
The representation $\Ad: \mathbb{G}\to \mathbb{GL}(\g)$ associated to the adjoint action $\Ad$ is 
\begin{equation*}
		\vcenter{\xymatrix{
			G_1 \ar[d]^{\Phi} \ar[r]^{\Ad^1} &\GL_1(\g)\ar[d]^{\Psi}\\
			G_0\ar[r]^{\Ad^0} &\GL_0(\g)		}},
	\end{equation*}
where $\Ad^0$ and $\Ad^1$ are given by
\begin{eqnarray}\label{adjoint1}	\Ad^0_{g_0}(x,u)=(\Ad^{G_0}_{g_0} x, g_0\triangleright u),\qquad \Ad^1_{g_1}(x)=-R_{g_1^{-1}*} \widehat{x}_{g_1}.
\end{eqnarray}
It is called the {\bf adjoint representation} of $\mathbb{G}$ on $\g$.
\end{Cor}



Taking the Lie 2-group dual as in Proposition \ref{dual}, we shall obtain the coadjoint representation of  $\mathbb{G}$ on $\g^*$ and furthermore a dual Lie 2-group action of $\mathbb{G}$ on $\g^*$. 
\begin{Cor}\label{cmcoad}
\begin{itemize} 
\item[\rm (i)] The coadjoint representation $\Ad^*: \mathbb{G}\to \mathbb{GL}(\g)$ of a Lie 2-group $\mathbb{G}$ on $\g^*$ is
\begin{equation*}
		\vcenter{\xymatrix{
			G_1 \ar[d]^{\Phi} \ar[r]^{\Ad^{*1}} &\GL_1(\g^*)\ar[d]^{\Psi}\\
			G_0\ar[r]^{\Ad^{*0}} &\GL_0(\g^*)		}},
	\end{equation*}
where $\Ad^{*0}$ and $\Ad^{*1}$ are given by
\begin{eqnarray*}\label{2groupactioncoad1}
	\Ad^{*0}_{g_0}(\xi,\alpha)=(g_0\triangleright^* \xi,\Ad^{*G_0}_{g_0} \alpha),\qquad \Ad^{*1}_{g_1}(\xi)=-\rho^*_{g_1}\xi,\qquad \forall \alpha\in \g_0^*,\xi\in \g_1^*.
	\end{eqnarray*}
	Here $\rho^*_{g_1}: \g_1^*\to \g_0^*$ is given by $\langle \rho^*_{g_1}(\xi),x\rangle=\langle \xi,R_{g_1*}\widehat{x}_{g_1^{-1}}\rangle$ for $x\in \g_0$.
	\item[\rm (ii)] There is a Lie 2-group action $\mathbb{A} \mathfrak{d}^*=\Ad^{2gp*}: \mathbb{G}\times \g^*\to \g^*$ of $\mathbb{G}$ on $\g^*$ given by
\begin{eqnarray}\label{2groupactioncoad}
\mathbb{A} \mathfrak{d}^*_{(g_0,g_1)}(\xi,\alpha)=(g_0\triangleright^*\xi, -\Ad_{g_0}^{*G_0}\rho_{g_1}^*\xi
+\Ad_{g_0\Phi(g_1)}^{*G_0} \alpha).
\end{eqnarray}
\end{itemize}
\end{Cor}

By Propositions \ref{sd2gp}, associated to  $\mathbb{A} \mathfrak{d}^*$, we have a semi-direct product Lie 2-group $\mathbb{G}\ltimes_{\mathbb{A}\mathfrak{d}^*} \g^*$, 
which is called the {\bf cotangent 2-group} of a Lie 2-group.

\begin{Rm}
Unlike the coadjoint orbits of a Lie group, the coadjoint orbits of a Lie 2-group defined by \eqref{2groupactioncoad} in general do not have natural symplectic structures. In fact, 
the coadjoint action $\Ad^*: \mathbb{G}\times \g^*\to \g^*$ of the Lie group $\mathbb{G}$ takes the formula
\begin{eqnarray*}
\label{coadjoint} \Ad^{*}_{(g_0,g_1)} (\xi,\alpha)&=&\big(g_0\triangleright^*(\Ad_{g_1}^*\xi),-\Ad_{g_0} ^*\rho_{g_1}^*\xi+\Ad_{g_0}^{*} \alpha\big),
\end{eqnarray*}
for any  $(\alpha,\xi)\in \g_0^*\oplus \g_1^*$. It is  different from \eqref{2groupactioncoad} and in general not a Lie 2-group action. By Remark \ref{sn}, it is a Lie 2-group action if and only if 
\[\Ad^{G_1}_{g_1} u=u,\qquad \Ad^{G_0}_{\Phi(g_1)} x=x,\qquad \forall g_1\in G_1, u\in \g_1,x\in \g_0.\]
So if $G_0$ and $G_1$ are connected, $\Ad^*$ is a Lie 2-group action if and only if $G_1$ is abelian and $\mathrm{Im}(\Phi)\subset Z(G_0)$.
\end{Rm}

\section{Homogeneous spaces of Lie 2-groups}

In this section, we give a thorough discussion of homogeneous spaces of Lie 2-groups. First, we
characterize the algebraic and geometric structures of the homogeneous $\mathbb{G}$-space $\mathbb{G}/\mathbb{H}$, where $\mathbb{H}$ is a closed 2-subgroup. It not only 
inherits  a canonical quotient Lie groupoid structure, but also is equipped with a fiber bundle structure. 
Then it is shown that a homogeneous $\mathbb{G}$-space is always of the quotient of $\mathbb{G}$ by a closed  2-subgroup. In the end, we discover that $\mathbb{G}/\mathbb{H}$ is not just a homogeneous space of the Lie group $\mathbb{G}$, but also a homogeneous space of some groupoid.

Given a Lie 2-group $\mathbb{G}$, we shall use $\diamond$ and $*$ to denote the group and groupoid structures on $\mathbb{G}$, respectively.

\begin{Def}
Let $\mathbb{G}:G_1\xrightarrow{\Phi} G_0$ be a Lie 2-group. A {\bf homogeneous space} of $\mathbb{G}$, or a {\bf homogeneous $\mathbb{G}$-space}, is a Lie groupoid $P\rightrightarrows P_0$ with a transitive $\mathbb{G}$-action. 
\end{Def}

\begin{Pro}\label{algstr}
Given a Lie 2-group $\mathbb{G}: G_1\xrightarrow{\Phi}G_0$ and a closed Lie 2-subgroup $\mathbb{H}: H_1\xrightarrow{\Phi} H_0$,  the quotient space $\mathbb{G}/\mathbb{H}$ is 
a homogeneous $\mathbb{G}$-space with the $\mathbb{G}$-action $g\triangleright [g']=[gg']$, where the Lie groupoid structure  on 
\[G_0\ltimes G_1/H_0\ltimes H_1\rightrightarrows G_0/H_0\]
is given as follows:  the source, target maps are 
\[s([g_0,g_1])=[g_0],\qquad t([g_0,g_1])=[g_0\Phi(g_1)],\]
and  the multiplication is 
\begin{eqnarray}\label{multquo}
[g_0,g_1]*[g_0\Phi(g_1),g'_1]=[g_0,g_1g'_1].
\end{eqnarray}
\end{Pro}
\begin{proof}
To see $\mathbb{G}/\mathbb{H}$ is a Lie groupoid, we only verify that the multiplication is well-defined. For any composable pair $[g_0,g_1],[l_0,l_1]\in \mathbb{G}/\mathbb{H}$, the condition  $t([g_0,g_1])=s([l_0,l_1])$ implies that
$[g_0\Phi(g_1)]=[l_0]$. So there exists $h_0\in H_0$ such that $g_0\Phi(g_1)=l_0h_0$. We can then write 
\[[l_0,l_1]=[l_0h_0,h_0^{-1}\triangleright l_1]=[g_0\Phi(g_1), h_0^{-1}\triangleright l_1].\]
This shows that to determine $*$ on $\mathbb{G}/\mathbb{H}$, it is enough to define multiplications of elements of $[g_0,g_1]$ and  $[g_0\Phi(g_1),g'_1]$ .
Then
for $(h_0,h_1), (h'_0,h'_1)\in H_0\ltimes H_1$, we have 
\begin{eqnarray*}
&&[(g_0,g_1) \diamond(h_0,h_1)]*[(g_0\Phi(g_1),g'_1)\diamond (h'_0,h'_1)]\\ &=&
[g_0h_0,(h_0^{-1}\triangleright g_1)h_1]*[g_0\Phi(g_1)h'_0,(h'^{-1}_0\triangleright g'_1)h'_1]\\ &=&[g_0h_0,(h_0^{-1}\triangleright g_1)h_1]*[g_0\Phi(g_1)h_0\Phi(h_1),(\Phi(h_1)^{-1}h^{-1}_0h'_0)\triangleright ((h'^{-1}_0\triangleright g'_1)h'_1)]
\\&=&[g_0h_0,(h_0^{-1}\triangleright g_1)h_1h_1^{-1}(h_0^{-1}\triangleright g'_1)((h_0^{-1}h'_0)\triangleright h'_1) h_1]\\ &=&[(g_0,g_1g'_1)\diamond \big(h_0,((h_0^{-1}h'_0)\triangleright h'_1) h_1\big)]\\&=&[g_0,g_1g'_1].
\end{eqnarray*}
This implies that the multiplication \eqref{multquo} does not depend on the choices of representatives in each equivalence class, so it is well-defined. It is direct to show that $\mathbb{G}/\mathbb{H}$ is a homogeneous $\mathbb{G}$-space, among which \eqref{goidmor} holds for the reason that the groupoid multiplication on $\mathbb{G}$ is a group homomorphism. 
\end{proof}

Given a Lie 2-group $\mathbb{G}$ with a closed Lie 2-subgroup $\mathbb{H}$, there is  a  natural Lie group  action of  $H_0$ on $G_1/H_1$ by 
\[H_0\times (G_1/H_1)\to G_1/H_1,\qquad h_0\triangleright [g_1]=[h_0\triangleright g_1].\]
We thus obtain an associated bundle 
\[G_0\times_{H_0} (G_1/H_1)\to G_0/H_0,\qquad [g_0,[g_1]]\mapsto [g_0]\]
of the principal $H_0$-bundle $G_0\to G_0/H_0$ with fiber $G_1/H_1$.
\begin{Lem}
The above associated bundle  is a Lie groupoid
$G_0\times_{H_0} (G_1/H_1)\rightrightarrows G_0/H_0$, where the source, target maps are 
\[s([g_0,[g_1]])=[g_0],\qquad t([g_0,[g_1]])=[g_0\Phi(g_1)],\]
and the multiplication is given by
\[[g_0,[g_1]]* [g_0\Phi(g_1),[g'_1]]=[g_0,[g_1g'_1]].\]
\end{Lem}
\begin{proof}
We first show the source and target maps are well-defined. Indeed, taking  $[g_0,[g_1]]=[g_0h_0,[(h_0^{-1}\triangleright g_1) h_1]]$ for any $(h_0,h_1)\in H_0\times H_1$, we have 
$s([g_0h_0,[(h_0^{-1}\triangleright g_1) h_1]])=[g_0h_0]=[g_0]$ and 
\[t([g_0h_0,[(h_0^{-1}\triangleright g_1) h_1]])=[g_0h_0h_0^{-1}\Phi(g_1)h_0\Phi(h_1)]=[g_0\Phi(g_1)h_0\Phi(h_1)]=[g_0\Phi(g_1)].\]
Then we show the multiplication is well-defined. For any $[g_0,[g_1]],[l_0,[l_1]]\in G_0\times_{H_0} (G_1/H_1)$ such that $t([g_0,[g_1]])=s([l_0,[l_1]])$, there exists some $h_0\in H_0$ such that 
$g_0\Phi(g_1)=l_0h_0$. So we have $[l_0,[l_1]]=[g_0\Phi(g_1),[h_0^{-1}\triangleright l_1]]$. This justifies that it is reasonable to only define the multiplication of elements $[g_0,[g_1]]$ and $[g_0\Phi(g_1),[g'_1]]$. It is straightforward to 
show the structures maps $s,t,*$ define a Lie groupoid structure.
\end{proof}
\begin{Pro}\label{geostr}
Let $\mathbb{G}$ be a Lie 2-group and $\mathbb{H}$ a closed Lie 2-subgroup. Then \begin{itemize}
\item[\rm (i)]  the Lie groupoid $G_0\times_{H_0} (G_1/H_1)\rightrightarrows G_0/H_0$ is a homogeneous $\mathbb{G}$-space with the $\mathbb{G}$-action  \[(l_0,l_1)\triangleright [g_0,[g_1]]=[l_0g_0,[(g_0^{-1}\triangleright l_1)g_1]];\] 
\item[\rm (ii)] we have an isomorphism of 
 homogeneous $\mathbb{G}$-spaces 
 \[\theta: G_0\ltimes G_1/H_0\ltimes H_1\to G_0\times_{H_0} (G_1/H_1),\qquad [g_0,g_1]\mapsto [g_0,[g_1]].\]\end{itemize}
\end{Pro}
\begin{proof}
(i) is direct to check.
For (ii), we first show that $\theta$ is well-defined. Actually,
\begin{eqnarray*}
\theta\big( [(g_0,g_1)\diamond(h_0,h_1)]\big)&=&\theta\big([g_0h_0,(h_0^{-1}\triangleright g_1) h_1]\big)=
[g_0h_0,[(h_0^{-1}\triangleright g_1) h_1]]\\ &=&[g_0h_0,h_0^{-1}\triangleright [g_1]]=[g_0,[g_1]]= \theta\big([g_0,g_1]\big).\end{eqnarray*}
 Then following from the calculation
 \begin{eqnarray*}
\theta([g_0,g_1]*[g_0\Phi(g_1),g'_1])&=&\theta([g_0,g_1g'_1])=[g_0,[g_1g'_1]]\\ &=&[g_0,[g_1]]* [g_0\Phi(g_1),[g'_1]]=\theta([g_0,g_1])*\theta([g_0\Phi(g_1),g'_1]),
\end{eqnarray*}
we see that $\theta$ is a Lie groupoid homomorphism.
Next, we check that $\theta$ is equivariant with the two $\mathbb{G}$-actions. Indeed, we have
\begin{eqnarray*}
\theta\big((l_0,l_1)\triangleright [g_0,g_1]\big)=\theta([l_0g_0,(g_0^{-1}\triangleright l_1) g_1])=[l_0g_0,[(g_0^{-1}\triangleright l_1) g_1]]=(l_0,l_1)\triangleright\theta([g_0,g_1]).
\end{eqnarray*}
At last, define a map 
\[\tau: G_0\times_{H_0} (G_1/H_1)\to G_0\ltimes G_1/H_0\ltimes H_1,\qquad [g_0,[g_1]]\mapsto [g_0,g_1].\]
It is straightforward to check that $\tau$ is well-defined and  it is the inverse of $\theta$. 
\end{proof}

The Lie groupoid $\mathbb{G}: G_0\ltimes G_{-1}\rightrightarrows G_0$ is an action groupoid relative to the right action of $G_1$ on $G_0$: $g_0\triangleleft g_1=g_0\Phi(g_1)$. As its quotient,  the groupoid $G_0\times _{H_0} (G_1/H_1)\cong \mathbb{G}/\mathbb{H}$ over $G_0/H_0$ is in general not an action groupoid any more. We give an example when it is still an action Lie groupoid.  Recall from \cite{N} that the subcrossed module $\mathbb{H}$ of a crossed module $\mathbb{G}$ is {\bf normal} if (1) $H_0\subset G_0$ is a normal subgroup; (2) $H_1$ is a $G_0$-submodule; (3) $(h_0\triangleright g_{1})
g_1^{-1}\in H_1$ for all $h_0\in H_0, g_1\in G_1$. If $\mathbb{H}\subset \mathbb{G}$ is a normal subcrossed module, then the quotient  $\mathbb{G}/\mathbb{H}: G_1/H_1\xrightarrow{\Phi} G_0/H_0$ inherits a crossed module structure such that the projection $\mathbb{G}\to \mathbb{G}/\mathbb{H}$ is a crossed module homomorphism.
\begin{Ex} \label{trivialbundle}
Let $\mathbb{G}$ be Lie 2-group and $\mathbb{H}$ a closed normal  2-subgroup (normal subcrossed module). Then 
\begin{itemize}
\item[\rm (i)] the quotient Lie groupoid 
$\mathbb{G}/\mathbb{H}$ is naturally isomorphic to the action Lie groupoid (Lie 2-group) $(G_0/H_0)\ltimes (G_1/H_1)\rightrightarrows G_0/H_0$.
\item[\rm (ii)] With the action $\triangleright: (G_0\ltimes G_1) \times (G_0/H_0\times G_1/H_1)\to G_0/H_0\times G_1/H_1$:
\[(g_0,g_1)\triangleright ([g'_0],[g'_1])=([g_0g'_0],[({g'_0}^{-1} \triangleright g_1)g'_1]),\]
the Lie groupoid $(G_0/H_0)\ltimes (G_1/H_1)\rightrightarrows G_0/H_0$ is a homogeneous $\mathbb{G}$-space such that the isomorphism in (i) is an isomomorphism of homogeneous $\mathbb{G}$-spaces.
\end{itemize}
\end{Ex}

\begin{Rm} 
As for the quotient Lie groupoid $G_0\times _{H_0} (G_1/H_1)\rightrightarrows G_0/H_0$,  there are two similar constructions:
\begin{itemize}
\item[\rm (1)] The first one is the quotient of an action Lie groupoid by Lu  \cite{Lu} and Blohmann-Weinstein \cite{BW}. A quotient of an action Lie groupoid by a subgroup is of the form $M\times_H G/H \rightrightarrows M/H$, where $M$ is a manifold with a Lie group $G$-action and $H\subset G$ is a subgroup.
Taking $G_0=G_1$ and $H_0=H_1$, our quotient Lie groupoid $G_0\times _{H_0} (G_0/H_0)\rightrightarrows G_0/H_0$ is an example of this case;
\item[\rm (2)] The second one is the quotient  of a PBG-Lie groupoid by Mackenzie \cite{Mackenzie}.  A  PBG-Lie groupoid is a locally trivial Lie groupoid $\Gamma\rightrightarrows P(M,G)$ over a principal bundle $P(M,G)$ together with a right action of $G$ on $\Gamma$ by groupoid automorphisms over the principal action $P\times G\to P$. Then the quotient $\Gamma/G\rightrightarrows M$ is a Lie groupoid. 
If $H_1$ is trivial, our quotient Lie groupoid $G_0\times _{H_0} G_1\rightrightarrows G_0/H_0$ is an example of Mackenzie's case.\end{itemize}

\end{Rm}
The Lie algebroid  $G_0\ltimes_{H_0} (\g_1/\h_1)\to G_0/H_0$ of the quotient Lie groupoid $ G_0\times _{H_0} (G_1/H_1)\rightrightarrows G_0/H_0$ can also be written down following from the construction in  \cite{Lu}. We omit the details.

It is from Proposition \ref{algstr} that the quotient of a Lie 2-group $\mathbb{G}$ by a closed  Lie 2-subgroup is  a homogeneous $\mathbb{G}$-space. We shall see that
a homogeneous $\mathbb{G}$-space $P\rightrightarrows P_0$ is always of this form. 
Given an action of a Lie 2-group $\mathbb{G}$  on a Lie groupoid $P\rightrightarrows P_0$,  for $x\in P_0$, denote by $\mathcal{O}^\mathbb{G}_x\subset P$ and $\mathcal{O}^{G_0}_x\subset P_0$ the $\mathbb{G}$-orbit and $G_0$-orbit passing through $x$, respectively. Let  $G_0^x\subset G_0$ be the isotropy subgroup of $x$ relative to the $G_0$-action on $P_0$ and let  $G_1^x\subset G_1$ and $\mathbb{G}^x\subset \mathbb{G}$ be the isotropy subgroups of $x$ with respect to the $G_1$ and $\mathbb{G}$-actions on $P$, respectively.

\begin{Thm} With the above notations,
\begin{itemize}
\item[\rm (i)] $\mathcal{O}^\mathbb{G}_x\rightrightarrows \mathcal{O}^{G_0}_x$ is a  Lie subgroupoid
of $P$.
With the induced $\mathbb{G}$-action,  it is a  homogeneous $\mathbb{G}$-space. Denote it by $\mathcal{O}_x^\mathbb{G}$.\item [\rm (ii)] 
$G_1^x\xrightarrow{\Phi} G_0^x$ is a subcrossed module of $G_1\xrightarrow{\Phi} G_0$ and $\mathbb{G}^x=G_0^x\ltimes G_1^x$. And the two homogeneous $\mathbb{G}$-spaces $\mathbb{G}/\mathbb{G}^x$ and $\mathcal{O}_x^\mathbb{G}$ are isomorphic.
\item[\rm (iii)] If $P$ is further a homogeneous $\mathbb{G}$-space, we have an isomorphism $P\cong  \mathbb{G}/\mathbb{G}^x$ of homogeneous $\mathbb{G}$-spaces.
 \end{itemize}
\end{Thm}
\begin{proof} For (i), for $g\triangleright x, l\triangleright x\in \mathcal{O}^\mathbb{G}_x$ with $t(g)\triangleright x=s(l)\triangleright x$, there exists $k\in G_0^x$ such that $t(g)=s(l)k=s(l\diamond k)$. Then
 by \eqref{goidmor}, we have 
 \[(g\triangleright x)*(l\triangleright x):=(g*_\mathbb{G} (l\diamond k))\triangleright x\in \mathcal{O}^\mathbb{G}_x,\]
which implies that $\mathcal{O}^\mathbb{G}_x$ is a Lie subgroupoid of $P$.

For (ii), first we show that $G_1^x\xrightarrow{\Phi} G_0^x$ is a subcrossed module of $\mathbb{G}$. For $g_1\in G_1^x$, since $\Phi(g_1)\triangleright x=t(g_1)\triangleright x=t_P(g_1\triangleright x)=t_P(x)=x$, we have $\Phi(g_1)\in G_0^x$. Also, for $
g_0\in G_0^x$ and $g_1\in G_1^x$, we have
\[(g_0\triangleright g_1)\triangleright x=(g_0\diamond g_1\diamond g_0^{-1})\triangleright x=x,\]
where $\diamond$ is the group multiplication in $\mathbb{G}$. Hence $g_0\triangleright g_1\in G_1^x$. Thus $G_1^x\xrightarrow{\Phi} G_0^x$ is a subcrossed module and gives a Lie 2-subgroup $G_0^x\ltimes G_1^x\rightrightarrows G_0^x$. 
To see $\mathbb{G}^x=G^x_0\ltimes G_1^x$, we first notice that $G^x_0\ltimes G_1^x\subset \mathbb{G}^x$. Then for an arbitrary $(g_0,g_1)\in \mathbb{G}^x$, we have $(g_0,g_1)\triangleright x=(g_0\diamond g_1)\triangleright x=x$. Applying the source map on both sides, we obtain
$g_0\triangleright x=x$ and further $g_1\triangleright x=g_0^{-1}\triangleright x=x$.  Hence $g_0\in G_0^x, g_1\in G_1^x$. Thus we proved $\mathbb{G}^x\subset G_0^x\ltimes G_1^x$. So they are equal as Lie groups. 

Then we show the map
\[F: \mathcal{O}_x^{\mathbb{G}}\to \mathbb{G}/\mathbb{G}^x,\qquad g\triangleright x\mapsto [g]\]
is an isomorphism of homogeneous $\mathbb{G}$-spaces. Indeed, it is obviously well-defined. To see it is a groupoid morphism, we find 
\[F\big((g\triangleright x)*(l\triangleright x)\big)=F\big((g*_\mathbb{G} (l\diamond k))\triangleright x\big)=[(g*_\mathbb{G} (l\diamond k)]=[g]* [l]=F(g\triangleright x)*F(l\triangleright x),\]
where we have used the notations in (i).
Also, $F(l\triangleright (g\triangleright x))=[l\diamond g]=l\triangleright[g]=l\triangleright F(g\triangleright x)$, we see $F$ is an isomorphism of homogeneous spaces with inverse $[g]\mapsto g\triangleright x$.

For (iii), by the transitivity condition, we have $P=\mathcal{O}_x$ for any $x\in P_0$. Then (iii) follows from (ii).
\end{proof}

\begin{Rm}
From this theorem, we see that for a Lie 2-group action, orbits passing through points in the base manifold are always homogeneous $\mathbb{G}$-spaces.
While  the orbit passing through a general point in $\mathbb{P}$ is not necessarily a homogeneous $\mathbb{G}$-space, since the isotropy group is in general not a Lie 2-subgroup. 
\end{Rm}

In this end of this section, we point out that the quotient space $\mathbb{G}/\mathbb{H}$ is not only a homogeneous space of a Lie group, but also a homogeneous space of a Lie groupoid. This fact will be used in the next section.

Let $\mathbb{G}: G_1\xrightarrow{\Phi} G_0$ be a Lie 2-group  and $\mathbb{H}: H_1\xrightarrow{\Phi} H_0$ a Lie 2-subgroup. Then the Lie group $H_0$ acts on $G_1$ and $H_1$, respectively, which induce two associated bundles
\[G_0\times _{H_0} G_1\to G_0/H_0,\qquad G_0 \times_{H_0} H_1\to G_0/H_0.\]
Taking $H_1$  as $\{e\}$ in Propositions \ref{algstr} and \ref{geostr}, we get that the associated bundle $G_0 \times _{H_0} G_1$ is a Lie groupoid 
\[G_0 \times _{H_0} G_1\rightrightarrows G_0/H_0\]
with source, target and multiplication
\[s([g_0,g_1])=[g_0],\qquad t([g_0,g_1])=[g_0\Phi(g_1)],\qquad [g_0,g_1]*[g_0\Phi(g_1),l_1]=[g_0,g_1 l_1].\]
Moreover, it is isomorphic to the Lie groupoid  $(G_0\ltimes G_1)/H_0\rightrightarrows G_0/H_0$.
\begin{Pro}\label{Gamma} With the notations above, we have 
\begin{itemize}
\item[\rm (i)] $G_0 \times_{H_0} H_1\rightrightarrows G_0/H_0$ is a Lie group bundle and a wide normal subgroupoid of $G_0 \times_{H_0} G_1$. 
Thus, the groupoid quotient $(G_0\times _{H_0} G_1)/\!\!/(G_0\times _{H_0} H_1)$ is a  Lie groupoid; 
\item[\rm (ii)] there is a Lie groupoid isomorphism 
\[\theta: (G_0\times _{H_0} G_1)/\!\!/(G_0\times _{H_0} H_1)\to  \mathbb{G}/\mathbb{H},\qquad [\![[g_0,g_1]]\!]\to [g_0,g_1].\]
\end{itemize}
Therefore,  $\mathbb{G}/\mathbb{H}$ is a homogeneous space of the Lie group $\mathbb{G}$ and also a homogeneous space of the Lie groupoid $\mathbb{G}/H_0\rightrightarrows G_0/H_0$.

\end{Pro}
\begin{proof}
For (i), it is obvious that $G_0 \times_{H_0} H_1\subset G_0 \times_{H_0} G_1$ is a wide subgroupoid. 
Following from $s([g_0,h_1])=t([g_0,h_1])=[g_0]$,  we have $G_0 \times_{H_0} H_1$ is a Lie group bundle over $G_0/H_0$. 
For $[g_0,h_1]\in G_0 \times_{H_0} H_1$ and $[g_0,g_1]\in G_0 \times_{H_0} G_1$, we have 
\[[g_0,h_1]*[g_0,g_1]=[g_0,h_1]*[g_0\Phi(h_1),\Phi(h_1)^{-1}\triangleright g_1]=[g_0,g_1h_1]=[g_0,g_1]*[g_0\Phi(g_1),h_1],\]
which implies that 
\[[g_0,g_1]^{\dagger}*[g_0,h_1]*[g_0,g_1]=[g_0\Phi(g_1),h_1]\in G_0 \times_{H_0} H_1,\] where $[g_0,g_1]^{\dagger}$ is the groupoid inverse of $[g_0,g_1]$. Thus $G_0 \times_{H_0} H_1\subset G_0 \times_{H_0} G_1$ is a wide normal subgroupoid.
Therefore, 
$(G_0\times _{H_0} G_1)/\!\!/(G_0\times _{H_0} H_1)\rightrightarrows G_0/H_0$ with the multiplication
\[[\![\gamma]\!]*[\![\gamma']\!]:=[\![\gamma*\gamma']\!],\qquad \forall (\gamma,\gamma')\in (G_0\times _{H_0} G_1)^{(2)}\]
is a Lie groupoid. For (ii), as 
\begin{eqnarray*}
\theta\big([\![[g_0h_0,h_0^{-1}\triangleright g_1]*[g_0\Phi(g_1), h_1]]\!]\big)&=&\theta([\![[g_0h_0,(h_0^{-1}\triangleright g_1)h_1]]\!])\\ &=&[g_0h_0,(h_0^{-1}\triangleright g_1)h_1]\\ &=&[(g_0,g_1)\diamond (h_0,h_1)]\\ &=&[g_0,g_1]=\theta([\![[g_0,g_1]]\!]),
\end{eqnarray*}
we see that $\theta$ is well-defined. Next, for $[g_0,g_1],[g_0\Phi(g_1),l_1]\in G_0\times _{H_0} G_1$, we have
\[\theta\big([\![[g_0,g_1]]\!]*[\![[g_0\Phi(g_1),l_1]]\!]\big)=\theta([\![[g_0,g_1l_1]]\!])=[g_0,g_1l_1]=[g_0,g_1]*[g_0\Phi(g_1),l_1],\]
which implies that $\theta$ is a homomorphism. It is moreover an isomorphism with the inverse given by $[g_0,g_1]\mapsto [\![[g_0,g_1]]\!]$.
\end{proof}

\section{Poisson homogeneous space of Poisson 2-groups}

Drinfeld classified Poisson homogeneous spaces of a Poisson Lie group by Dirac structures of its Lie bialgebra \cite{D2}, 
and then Liu-Weinstein-Xu generalized this result to Poisson groupoids  in \cite{LWX}. Here we present the classification theorem for Poisson homogeneous spaces of Poisson 2-groups. Let us first collect some standard definitions and facts regarding Lie bialgebras and Poisson homogeneous spaces.

A {\bf Lie bialgebra} is a pair $(\g,\g^*)$ of Lie algebras such that
\[d_*[x,y]=[d_*(x),y]+[x,d_*(y)],\qquad \forall x,y\in \g,\]
where $d_*:\g\to \wedge^2 \g$ is defined by $\langle d_*(x),\xi\wedge \eta\rangle:=-\langle x,[\xi,\eta]_{\g^*}\rangle$.
The double $\g\oplus \g^*$  of a Lie bialgebra is associated with a natural pairing:
\[\langle x+\xi,y+\eta\rangle_{\bowtie} =\langle x,\eta\rangle+\langle y,\xi\rangle,\qquad \forall x,y\in  \g, \xi,\eta\in \g^*,\]
 and a Lie bracket:
 \begin{eqnarray}\label{doublebr}
[x+\xi,y+\eta]_{\bowtie}=[x,y]_\g+\ad_\xi^* y-\ad_\eta^* x+[\xi,\eta]_{\g^*}+\ad^*_x \eta-\ad_y^* \xi.
\end{eqnarray}
Denote the double Lie algebra by $\g\bowtie \g^*$.
A {\bf Dirac structure} of a Lie bialgebra $(\g,\g^*)$ is a subspace $L\subset \g\oplus \g^*$  which is maximal isotropic with respect to $\langle \cdot,\cdot\rangle_{\bowtie}$ and  closed under the Lie bracket $[\cdot,\cdot]_{\bowtie}$. 
The notion of characteristic pairs was introduced in \cite{L} as a useful tool  to characterize Dirac structures. Indeed,
a maximal isotropic subspace of $\g\bowtie \g^*$ is always of the form
\begin{eqnarray}\label{Lliealg}
L=\{x+\Omega^\sharp \xi+\xi; x\in \h,\xi\in \h^\perp\}=\h\oplus \mathrm{Graph}(\Omega|_{\h^\perp}),
\end{eqnarray}
where $\h\subset \g$ is a subspace,  $\Omega\in \wedge^2 \g$ and $\Omega^\sharp(\xi):=\Omega(\xi,\cdot)$. Note that $\Omega$ is actually determined by $\mathrm{pr}(\Omega)\in \wedge^2 \g/\h$. 
We call the pair $(\h,\Omega)$ a {\bf characteristic pair} of $L$. 
A maximal isotropic subspace $L$ as in \eqref{Lliealg} is a Dirac structure of $(\g,\g^*)$ if and only if 
\begin{itemize}
\item[\rm(1)] $\h$ is a Lie subalgebra of $\g$;
\item[\rm (2)] $d_*\Omega+\frac{1}{2}[\Omega,\Omega]_\g\equiv 0$ (mod $\h$);
\item[\rm (3)] $[\xi,\eta]_{\g^*}+[\xi,\eta]_\Omega\in \h^\perp$ for all $\xi,\eta\in \h^\perp$, where $[\xi,\eta]_\Omega=\ad^*_{\Omega^\sharp \xi}\eta-\ad_{\Omega^\sharp \eta}^* \xi$.\end{itemize}
Here, the second condition means that $d_*\Omega+\frac{1}{2}[\Omega,\Omega]_\g\in \h\wedge \g\wedge \g$.


Let $(G,\pi)$ be a Poisson Lie group. A {\bf Poisson homogeneous $G$-space} is a Poisson manifold  $(M,\pi_M)$ such that $M$ is a homogeneous $G$-space and the action $G\times M\to M$ is a Poisson map. 
It is by Drinfeld that there is a one-one correspondence between regular Dirac structures $L\subset \g\oplus \g^*$ and Poisson homogeneous spaces $G/H$, where $H$ is a connected closed subgroup of $G$ with Lie algebra  $L\cap \g$.
A Dirac structure $L$ is regular if the left translation of $\h=L\cap \g$ defines a simple foliation on $G$.  


Let $(P,\Pi)$ be a Poisson groupoid. A $P$-space is homogeneous if and only if it is isomorphic to the groupoid quotient $P/\!\!/P_1$ for some wide (containing all the identities) subgroupoid $P_1\subset P$. If $P/\!\!/P_1$ is equipped with a Poisson structure such that the groupoid action is a Poisson map,  it is called a {\bf Poisson homogeneous space} of $(P,\Pi)$. 
Liu-Weinstein-Xu \cite{LWX} proved that there is a one-one correspondence between Poisson homogeneous spaces $P/\!\!/P_1$ and regular Dirac structures $L$ of the Lie bialgebroid $(A,A^*)$, where $P_1$ is the $s$-connected closed subgroupoid of $P$ corresponding to the subalgebroid $L\cap A$. 


A Poisson 2-group is both a Poisson Lie group and a Poisson groupoid. So a Poisson homogeneous space of a Poisson 2-group is naturally expected to be a Poisson homogeneous space of the Poisson Lie group and Poisson groupoid simultaneously. However, given a Lie 2-group $\mathbb{G}$ and a closed Lie 2-subgroup $\mathbb{H}$, the subgroupoid $H_0\ltimes H_1\rightrightarrows H_0$ is not a wide subgroupoid of $G_0\ltimes G_1\rightrightarrows G_0$ if $H_0\neq G_0$. 
To adjust to the framework in \cite{LWX}, we consider $\mathbb{G}/\mathbb{H}$ as a groupoid homogeneous space of $\mathbb{G}/H_0$ as in Proposition \ref{Gamma}.

\begin{Lem}\label{quotient}
Let $(\mathbb{G},\pi_{\mathbb{G}})$ be a Poisson 2-group. Suppose  $H_0\subset G_0$ is a closed subgroup. Then $(\mathbb{G}/H_0,\mathrm{pr}_*\pi_{\mathbb{G}})$  is a quotient Poisson groupoid  if and only if $H_0\subset \mathbb{G}$ is a Poisson subgroup.
\end{Lem}
\begin{proof}
We first recall two facts: the quotient space $\mathbb{G}/H_0$ inherits a Poisson structure such that the projection $\mathbb{G}\to \mathbb{G}/H_0$ is a Poisson map if and only if the annihilators $\h_0^\perp\oplus \g_1^*$  of $\h_0$ in $\g$ is a Lie subalgebra of $\g^*$. That is, $\h_0$ is a coideal  in $\g$; The subgroup $H_0\subset \mathbb{G}$ is a Poisson submanifold of $\mathbb{G}$  if and only if $\h_0^\perp\oplus \g_1^*$ is an ideal of $\g^*$.

Considering the Lie 2-algebra structure on $\g^*$, we note that $\h_0^\perp\oplus \g_1^*\subset \g^*$ is a Lie subalgebra if and only if it is an ideal, if and only if 
$\h_0^\perp\subset \g_0^*$ is a $\g_1^*$-submodule. 
So  with the quotient groupoid structure, $\mathbb{G}/H_0$ is a quotient Poisson groupoid.
\end{proof}

\begin{Def}
Let $(\mathbb{G},\Pi_{\mathbb{G}})$ be a Poisson Lie 2-group and  $\mathbb{H}\subset \mathbb{G}$  a closed 2-subgroup  such that $H_0\subset \mathbb{G}$ is a Poisson subgroup. The homogeneous $\mathbb{G}$-space $\mathbb{G}/\mathbb{H}$ is called a {\bf Poisson homogeneous $\mathbb{G}$-space} if there is a Poisson structure $\Pi$ such that  $(\mathbb{G}/\mathbb{H}, \Pi)$  is both a  Poisson homogeneous space of the Poisson Lie group $\mathbb{G}$ and a  Poisson homogeneous space of the Poisson groupoid $\mathbb{G}/H_0$. 
\end{Def}

For a  Poisson 2-group $(\mathbb{G},\pi_{\mathbb{G}})$ with Lie 2-bialgebra $(\g,\g^*)$, we see that $(\g,\g^*)=(\g_0\ltimes \g_1,\g_1^*\ltimes \g_0^*)$ is a Lie bialgebra and $(\g_0,\g_0^*), (\g_1,\g_1^*)$ are both Lie bialgebras; see \cite[Proposition 3.6]{CSX}.  But we emphasize that  the doubles $\g_i\bowtie \g_i^*\subset \g\bowtie \g^*$ for $i=0,1$ are  not Lie subalgebras. For example, for $x\in \g_0$ and $\alpha\in \g^*_0$, the bracket $[x,\alpha]_{\g\bowtie \g^*}$ also has component in $\g_1$, which is determined by $\langle [x,\alpha]_{\g\bowtie \g^*}, \xi\rangle=-\langle x, \xi\triangleright \alpha\rangle$ for $\xi\in \g_1^*$.

A Dirac structure $L$ of the Lie bialgebra $(\g_1,\g_1^*)$ is {\bf $H_0$-invariant} if $H_0\triangleright L=L$, where the action is the natural action of $H_0$ on $\g_1\oplus \g_1^*$.

\begin{Thm}\label{Main}
Let $(\mathbb{G},\Pi_{\mathbb{G}})$ be a Poisson Lie 2-group and  $\mathbb{H}\subset \mathbb{G}$  a closed 2-subgroup  such that $H_0\subset \mathbb{G}$ is a Poisson subgroup. Then there is a one-one correspondence between Poisson homogeneous $\mathbb{G}$-spaces $\mathbb{G}/\mathbb{H}$  and $H_0$-invariant Dirac structures $L$ of the Lie bialgebra $(\g_1,\g_1^*)$ such that $L\cap \g_1=\h_1$. 
\end{Thm}

We first characterize Dirac structures of the Lie bialgebra $(\g,\g^*)$. 
Let us  convent that $x,y\in \g_0,u,v\in \g_1$ and $\xi,\eta\in \g_1^*,\alpha,\beta\in \g_0^*$. 
A maximal isotropic subspace $L\subset \g\oplus \g^*$ such that $L\cap \g_i$ is a subspace of $\g_i$ for $i=0,1$ is always  of the form:
\begin{eqnarray}\label{mi}
\nonumber L&=&\{x+u+r^\sharp(\xi)+K^\sharp(\xi)+\pi^\sharp(\alpha)+K^\sharp(\alpha)+\xi+\alpha| x\in \h_0, u\in \h_1,\xi\in \h_1^\perp, \alpha\in \h_0^\perp\}\\
&=&\h\oplus \mathrm{Graph}(\Omega|_{\h^\perp}),
\end{eqnarray} 
where $\h_i\subset \g_i, i=0,1$ are subspaces and 
\[\Omega=\pi+K+r\in \wedge^2 \g_0\oplus  \g_0\wedge \g_1 \oplus \wedge ^2 \g_1.\]
Here $K^\sharp: \g_0^*\to \g_1$ and $K^\sharp: \g_1^*\to \g_0$ are defined by $K^\sharp(\alpha)=K(\alpha,\cdot)$ and $K^\sharp(\xi)=-K(\cdot,\xi)$.
Such a pair $(\h,\Omega)$ with $\h\subset \g $ and $\Omega\in \wedge^2 \g$ is  a characteristic pair of $L$.

\begin{Lem}\label{Hinv}
Let $L\subset \g_1\oplus \g_1^*$ be a Dirac structure of the Lie bialgebra $(\g_1,\g_1^*)$ with the characteristic pair $(\h_1,r\in \wedge^2\g_1)$. Then $L$ is $H_0$-invariant if and only if $H_0\triangleright \h_1\subset \h_1$ and $H_0\triangleright r\equiv r \pmod{\h_1}$.
\end{Lem}
\begin{proof}
Note that $L$ is $H_0$-invariant if $h_0\triangleright (x+r^\sharp \xi+\xi)\in L$ for $h_0\in H_0, x\in \h_1$ and $\xi\in \h_1^\perp$, which is equivalent to \[h_0\triangleright \h_1\subset \h_1,\qquad h_0\triangleright (r^\sharp \xi)-r^\sharp (h_0\triangleright^* \xi)\in \h_1.\]
The second equation amounts to the fact that $h_0 \triangleright r-r\in \h_1\wedge \g_1$.\end{proof}

\begin{Lem}\label{alcp}
Let $(\g,\g^*)$ be a Lie 2-bialgebra and assume $\h\subset \g$ is a Lie 2-subalgebra. 
Then a maximal isotropic subspace  $L\subset \g\oplus \g^*$ with characteristic pair $(\h,r\in \wedge^2\g_1)$ is a Dirac structure of the Lie bialgebra $(\g,\g^*)$  if and only if the characteristic pair $(\h_1, r)$ defines an $\h_0$-invariant Dirac structure $L_1$ of $(\g_1,\g_1^*)$ such that $\h_1^\perp\triangleright \h_0^\perp\subset \h_0^\perp$. 
\end{Lem}
\begin{proof}
First, following  from the general framework for Dirac structures of Lie bialgebras, we see that $L\subset \g\oplus \g^*$ with characteristic pair $(\h,r\in \wedge^2\g_1)$ is a Dirac structure of  $(\g,\g^*)$ if and only if 
\begin{itemize}
\item[\rm (i)] $\h \subset \g$ is a Lie subalgebra;
\item[\rm (ii)] $d_{*} r+\frac{1}{2}[r,r]\equiv 0 \pmod{\h_1}$;
\item[\rm (iii)] $[a,b]+[a,b]^\g_r\in \h^\perp$ for $a,b\in \h^\perp$, where $[a,b]^\g_r=\ad_{r^\sharp a}^{*}b-\ad_{r^\sharp b}^{*}a$ and $\ad^*$ is the coadjoint action of $\g$ on $\g^*$.
\end{itemize}
Note that $d_{*} r=d_{\g_{1}^*} r$ in (ii).  Then, we claim that (iii) is equivalent to the following conditions:
\begin{eqnarray*}\label{four}\qquad 
[\h_0^\perp,\h_0^\perp]\subset \h_0^\perp,\qquad \h_1^\perp\triangleright \h_0^\perp\subset \h_0^\perp,\qquad [\xi,\eta]+[\xi,\eta]_{r}\in \h_1^\perp,\qquad x\triangleright r \equiv 0 \pmod{\h_1},
\end{eqnarray*}
for all $x\in \h_0$ and $\xi,\eta\in \h_1^\perp $, where $[\xi,\eta]_r=\mathfrak{ad}_{r^\sharp \xi}^{*}\eta-\mathfrak{ad}^*_{r^\sharp \eta}\xi$ and $\mathfrak{ad}^*$ is the coadjoint action of $\g_1$ on $\g_1^*$.
 In fact, for $a,b\in \h_0^\perp$ and $a\in\h_0^\perp, b\in \h_1^\perp$,  we obtain the first two conditions. For $\xi,\eta\in \h_1^\perp$, we find that 
\[\mathrm{pr}_{\g_1^*}([\xi,\eta]+[\xi,\eta]^{\g}_r)=[\xi,\eta]+[\xi,\eta]_{r}\]
 and the $\g_0^*$-part is determined by
\[\langle \mathrm{pr}_{\g_{0}^*}([\xi,\eta]+[\xi,\eta]^{\g}_r),x\rangle=\langle \mathrm{pr}_{\g_{0}^*}([\xi,\eta]^\g_r),x\rangle
=\langle \eta, x\triangleright (r^\sharp \xi)\rangle-\langle \xi, x\triangleright (r^\sharp \eta)\rangle=\langle x\triangleright r,\xi\wedge \eta\rangle,\]
which implies $x\triangleright r\in \h_1\wedge \g_1$ for any $x\in \h_0$. 
Therefore, for $a,b\in \h_1^\perp$, the condition $[a,b]+[a,b]^{\g}_r\in \h^\perp$ infers the last two conditions.

Next, on the condition that $\h\subset \g$ is a Lie 2-subalgebra, $\h_1^\perp \triangleright \h_0^\perp\subset \h_0^\perp$ implies that $\h_0^\perp\subset \g_0^*$ is a Lie subalgebra. Thus, $L$  is  a Dirac structure of  $(\g,\g^*)$ if and only if 
\[\h_1^\perp \triangleright \h_0^\perp\subset \h_0^\perp, \quad d_{\g_1^*} r+\frac{1}{2}[r,r]\equiv 0 \pmod{\h_1},\quad  [\xi,\eta]+[\xi,\eta]_{r}\in \h_1^\perp,\quad x\triangleright r \equiv 0 \pmod{\h_1},\]
for $\xi,\eta\in \h_1^\perp$ and $x\in \h_0$, which implies that $L_1$ is an $\h_0$-invariant Dirac structure  of $(\g_1,\g_1^*)$ such that $\h_1^\perp\triangleright \h_0^\perp\subset \h_0^\perp$.
\end{proof}


This Lemma demonstrates that a Dirac structure $L\subset \g\bowtie \g^*$ with  characteristic pair $(\h,r)$ for $r\in \wedge^2 \g_{1}$ induces two Dirac structures $L_1\subset \g_1\bowtie \g_1^*$ and $\h_0\oplus \h_0^\perp \subset \g_0\bowtie \g_0^*$, respectively. Moreover, we have  $L=L_1\oplus (\h_0\oplus \h_0^\perp)$ as vector spaces.

\begin{Lem}\label{important}
Given a Poisson 2-group $(\mathbb{G},\Pi_{\mathbb{G}})$, let $L$ be a Dirac structure of the Lie bialgebra $(\g,\g^*)$ whose characteristic pair is $(\h,\Omega=\pi+K+r)$, then  $(\mathbb{G}/\mathbb{H},\Pi)$ is a Poisson homogeneous space of the Poisson group $\mathbb{G}$ with the action given by the group left translation and 
\[\Pi=\mathrm{pr}_*(\Pi_{\mathbb{G}}+\overleftarrow{\Omega}^{gp}),\]
where $\overleftarrow{\cdot}^{gp}$ denotes the left translation relative to the Lie group structure on $\mathbb{G}$. In fact, suppose
\begin{eqnarray}\label{Omega}
 \qquad \pi=\omega^{ij} e_i\wedge e_j\in \wedge^2 \g_0,\quad K=\theta^{ik}e_i\wedge \xi_k\in \g_{0}\wedge \g_1,\quad r=\sigma^{kl}\xi_k\wedge \xi_l\in \wedge^2 \g_1,
\end{eqnarray}
where $\{e_i\}$ and $\{\xi_k\}$ are bases of $\g_0$ and $\g_1$, respectively. Then 
\[\overleftarrow{\Omega}^{gp}=\Omega^{2,0}+\Omega^{1,1}+\Omega^{0,2}\in \mathfrak{X}^2(G_0\ltimes G_1)\]
is given by 
\begin{eqnarray*}
\Omega^{2,0}&=&\omega^{ij} \overleftarrow{e_i}\wedge \overleftarrow{e_j}\qquad\in \mathfrak{X}^2(G_0);\\
\Omega^{1,1}&=&-2\omega^{ij} \overleftarrow{e_i}\wedge \widehat{e_j}+\theta^{ik}\overleftarrow{e_i}\wedge \overleftarrow{\xi_k}\qquad \in \mathfrak{X}(G_0)\otimes \mathfrak{X}(G_1);\\
\Omega^{0,2}&=&\omega^{ij} \widehat{e_i}\wedge \widehat{e_j}-\theta^{ik}\widehat{e_i}\wedge \overleftarrow{\xi_k}+\sigma^{kl}\overleftarrow{\xi_k}\wedge \overleftarrow{\xi_l}\qquad \in \mathfrak{X}^2(G_1),
\end{eqnarray*}
where $\overleftarrow{e_i}\in \mathfrak{X}(G_0)$ and $\overleftarrow{\xi}\in \mathfrak{X}(G_1)$ are vector fields generated by left translations and $\widehat{e_i}\in \mathfrak{X}(G_1)$ is the fundamental vector field coming from the action of $G_0$ on $G_1$.
\end{Lem}
\begin{proof} The result follows from Drinfeld's classification theorem for Poisson homogeneous spaces in \cite{D2}. Here we compute the formula of $\overleftarrow{\Omega}^{gp}$.
For $x\in \g_0, u\in \g_1$, with respect to the group multiplication on $\mathbb{G}$, we have \[\overleftarrow{x}^{gp}(g_0,g_1) =(L_{g_0*} x,\frac{d}{dt}|_{t=0}(e^{-tx}\triangleright g_1)=(L_{g_0*} x, -\widehat{x}_{g_1}),\qquad \overleftarrow{u}^{gp}(g_0,g_1)=L_{g_1 *} u.\]
This shows that $\overleftarrow{e_i}^{gp}=\overleftarrow{e_i}-\widehat{e_i}\in TG_0\oplus TG_1$ and $\overleftarrow{\xi_k}^{gp}=\overleftarrow{\xi_k}$. \end{proof}

With the notations and assumptions in Lemma \ref{quotient}, denote by $(A,A^*)$ the Lie bialgebroid of the quotient Poisson groupoid $(\mathbb{G}/H_0,\mathrm{pr}_* \pi_{\mathbb{G}})=(G_0\times_{H_0} G_1,\mathrm{pr}_* \pi_{\mathbb{G}})$. Note that  $A=G_0\times_{H_0} \g_1\to G_0/H_0$ is a quotient of an action Lie algebroid and 
\[\Gamma(A)=\{\Omega'\in C^\infty(G_0, \g_1); \Omega'(g_0h_0)=h_0^{-1}\triangleright \Omega'(g_0)\}.\]
A maximal isotropic  subbundle $L\subset A\oplus A^*$ is of the form:
\begin{eqnarray}\label{cpla}
L=\{X+\Omega^\sharp \xi+\xi|X\in D,\xi\in D^\perp\}=D\oplus \mathrm{Graph}(\Omega|_{D^\perp}),
\end{eqnarray}
where $D\subset A$ is a subbundle, $D^\perp\subset A^*$ is the conormal bundle of $D$ and $\Omega\in \Gamma(\wedge^2 A)$. In fact, $\Omega$ is determined by its projection  $\mathrm{pr}(\Omega)\in \Gamma(\wedge^2 A/D)$. Such a pair $(D,\Omega)$ is called a characteristic pair of $L$.

For example, for any $r\in \wedge^2 \g_1$ such that $H_0\triangleright r\equiv r \pmod{\h_1}$,  the projection \[\mathrm{pr}(r)\in \wedge^2 (\g_{-1}/\h_{-1})^{H_0} \subset  \Gamma(\wedge^2 A/D)\] is a 2-section of $A/D$. Then any lifting 
$\widehat{\mathrm{pr}(r)}\in \Gamma(\wedge^2 A)$ gives rise to a 2-section of $A$.
\begin{Lem}\label{oidcp}
Let $(A,A^*)$ be the above Lie bialgebroid,  $L\subset A\oplus A^*$ a maximal isotropic subbundle with characteristic pair $(D=G_0\times _{H_0} \h_1,\widehat{\mathrm{pr}(r)})$, where $r\in \wedge^2 \g_1$ such that $H_0\triangleright r\equiv r \pmod{\h_1}$. 
Then $L$  is a Dirac structure of the Lie bialgebroid $(A,A^*)$  if and only if the following conditions hold:
\begin{itemize}
\item [\rm (1)] $\h_1\subset \g_1$ is a Lie subalgebra;
\item[\rm (2)] $d_{\g_1^*} r+\frac{1}{2}[r,r]\equiv 0 \pmod{\h_1}$; 
\item[\rm (3)] $[\xi,\eta]+[\xi,\eta]_r\in \h_1^\perp$ for any $\xi,\eta\in \h_1^\perp$. 
\end{itemize}
\end{Lem}
\begin{proof}Note that $(D=G_0\times_{H_0} \h_1,\widehat{\mathrm{pr}(r)})$ is a Dirac structure of the Lie bialgebroid $(A,A^*)$ if and only if $(G_0\times \h_1, r)$ is an $H_0$-invariant Dirac structure of the Lie bialgebroid $(G_0\times \g_1, G_0\times \g_1^*)$, if and only if $(\h_1, r)$ is a Dirac structure of the Lie bialgebra $(\g_1,\g_1^*)$. This proves the lemma. 
\end{proof}

Now we are ready to prove Theorem \ref{Main}.

\begin{proof}
Given a Poisson homogeneous space $(\mathbb{G}/\mathbb{H}, \Pi)$ of the Poisson 2-group $\mathbb{G}$, since it is a Poisson homogeneous space of the Poisson group $\mathbb{G}$, we have \[\Pi-\mathrm{pr}_*(\Pi_{\mathbb{G}})=\mathrm{pr}_*\overleftarrow{\Omega}^{gp}\] for some $\Omega\in \wedge^2 \g=(\wedge^2 \g_0)\oplus  (\g_0\wedge \g_1) \oplus (\wedge ^2 \g_1)$. Here we use $\overleftarrow{\cdot}^{gp}$ to denote the left translation with respect to the group multiplication. As $(\mathbb{G}/\mathbb{H}, \Pi)$ is also a Poisson homogeneous space of the Poisson groupoid $(\mathbb{G}/H_0,\mathrm{pr}_* \pi_{\mathbb{G}})$, we have 
\[\Pi-\mathrm{pr}_*(\Pi_{\mathbb{G}})=\mathrm{pr}_*\overleftarrow{\Omega'}^{gpd}\]
for 
\[\Omega'\in \Gamma(\wedge^2 A)=\{\Omega'\in C^\infty(G_0, 
\wedge^2 \g_1); \Omega'(g_0h_0)=h_0^{-1}\triangleright \Omega'(g_0)\},\]where $\overleftarrow{\cdot}^{gpd}$ denotes the left translation with respect to the groupoid multiplication.  
By Lemma \ref{important}, we have 
\[\overleftarrow{\Omega}^{gp}=\Omega^{2,0}+\Omega^{1,1}+\Omega^{0,2}\in \mathfrak{X}^2(G_0\ltimes G_1),\quad \overleftarrow{\Omega'}^{gpd}=\Omega'^{0,2}\in \mathfrak{X}^2(G_0\ltimes G_1).\]
So in order to make sure $\mathrm{pr}_*(\overleftarrow{\Omega}^{gp}-\overleftarrow{\Omega}^{gpd})=0$ holds,  we must have 
\[\Omega=r\in \wedge^2 \g_1 \pmod{\h_1}, \qquad \Omega'= \widehat{\mathrm{pr}(r)},\qquad \mathrm{and}\qquad  H_0\triangleright r\equiv r \pmod{\h_1}.\]
Here observe that $\overleftarrow{r}^{gpd}=\overleftarrow{r}^{gp}$.

Then applying Drinfeld theorems for Poisson homogeneous spaces of Poisson groups and groupoids, and by Lemmas \ref{Hinv}, \ref{alcp} and \ref{oidcp}, we obtain the one-one correspondence. To be precise, the characteristic pairs of the Dirac structures of the bialgebra $(\g,\g^*)$ and the Lie bialgebroid $(A,A^*)$ are $(\h_0\oplus \h_1, r)$ and $(G_0\times_{H_0} \h_1, \widehat{\mathrm{pr}(r)})$, respectively.
\end{proof}
Recall that a bivector field $\Pi$ on a Lie groupoid $P$ is {\bf multiplicative} if and only if the submanifold \[\Gamma=\{(g,h,gh); t(g)=s(h) \} \subset P\times P\times P\]
is coisotropic with respect to $\Pi\oplus \Pi\oplus (-\Pi)$. It is called {\bf affine} if and only if the submanifold \[\Lambda=\{(g,h,l, hg^{-1}l); t(h)=t(g), s(g)=s(l)\} \subset P\times P\times P\times P\]
is coisotropic with respect to $\Pi\oplus (-\Pi)\oplus (-\Pi)\oplus \Pi$. See \cite{W} for affine Poisson structures and \cite{LLS} for the general case. 
A Lie groupoid with an affine Poisson structure is called an {\bf affine Poisson groupoid}.
 
\begin{Cor}
With the assumptions in Theorem \ref{Main}, if $\h_1^\perp\subset \g_1^*$ is a Lie subalgebra, then Poisson homogeneous $\mathbb{G}$-spaces $\mathbb{G}/\mathbb{H}$  are  affine Poisson groupoids.
\end{Cor}
\begin{proof}
 If $H_0\subset \mathbb{G}$ is a Poisson subgroup, we see $\h_0^\perp\subset \g_0^*$ is a $\g_1^*$-submodule. Together with the condition $[\h_1^\perp,\h_1^\perp]\subset \h_1^\perp$, we have $\h^\perp=\h_0^\perp\oplus \h_1^\perp$ is a Lie subalgebra of $\g^*$, which leads to the fact that 
$(\mathbb{G}/\mathbb{H},\mathrm{pr}_*\Pi_{\mathbb{G}})$ is a Poisson manifold and it is obviously a Poisson groupoid. 
 
 It is known from \cite{LLS}  that $\Pi\in \mathfrak{X}^2(P)$ is multiplicative if and only if $\Pi+\overleftarrow{\lambda}^{gpd}$ for all $\lambda\in \Gamma(\wedge^2 A_P)$ is affine, where $A_P$ is the Lie algebroid of the Lie groupoid $P$.  By the proof of Theorem \ref{Main}, the Poisson structure on the Poisson homogeneous space $\mathbb{G}/\mathbb{H}$ is  \[\Pi=\mathrm{pr}_*(\Pi_{\mathbb{G}}+\overleftarrow{r}^{gp})=
 \mathrm{pr}_*\Pi_{\mathbb{G}}+\overleftarrow{\widehat{\mathrm{pr}(r)}}^{gpd},\qquad r\in \wedge^2 \g_1,\] which is affine. 
 \end{proof}
 
By \cite{CSX}, a Poisson 2-group $(G_0\ltimes G_1\rightrightarrows G_0,\Pi_{\mathbb{G}})$ has a mixed product Poisson structure (see \cite{LM}). Namely, both $G_0$ and $G_1$ are quotient Poisson manifolds. Moreover, $(G_1,\Pi_{G_1})$ is a Poisson subgroup of $\mathbb{G}$, where $\Pi_{G_1}=\Pi_{\mathbb{G}}|_{\mathbb{G}_1}$.
\begin{Cor}
Let  $\mathbb{G}/\mathbb{H}$  be a Poisson homogeneous space of the Poisson 2-group $\mathbb{G}$. Then $G_1/H_1$ is a Poisson submanifold of $\mathbb{G}/\mathbb{H}$ and it is a Poisson homogeneous space of the Poisson  group $G_1$. 
\end{Cor}
\begin{proof}
Suppose $\mathbb{G}/\mathbb{H}$ with the Poisson structure $\Pi= \mathrm{pr}_*(\Pi_{\mathbb{G}}+\overleftarrow{r})$ for $r\in \wedge^2 \g_1$  is a Poisson homogeneous space of the Poisson 2-group $(\mathbb{G},\Pi_{\mathbb{G}})$. 
Since $G_1$ is a Poisson subgroup of $\mathbb{G}$ and $r\in \wedge^2 \g_1$, we see that $(\h_1,r)$ is a Dirac structure of the Lie bialgebra $(\g_1,\g_1^*)$. By Drinfeld's theorem, 
we have that $(G_1/H_1,\mathrm{pr}_*(\Pi_{G_1}+\overleftarrow{r}))$ is a Poisson homogeneous space of the Poisson group $(G_1,\Pi_{G_1})$. Moreover, it is obvious that 
 the inclusion map 
\[G_1/H_1\hookrightarrow \mathbb{G}/\mathbb{H}, \qquad [g_1]\mapsto [e,g_1]\]
is a Poisson map.  
\end{proof}

Let us recall a class of  Lie 2-bialgebras coming from an $r$-matrix, called a {\bf coboundary Lie 2-bialgebra}. Let $\g$ be a Lie 2-algebra. An element  $\mu\in \wedge^2 \g_1$ is called an {\bf $r$-matrix} of $\g$ if $x\triangleright[\mu,\mu]=0$ for all $x\in \g_0$. Such a $\mu$ gives rise to a Lie 2-bialgebra $(\g,\g^*)$, where the Lie bracket on $\g_1^*$ is defined by
\[\langle [\xi,\eta]_\mu,u\rangle:=\langle\xi\wedge \eta, [\mu,u]\rangle,\qquad \forall \xi,\eta\in \g_1^*, u\in \g_1,\]
and the action of $\g_1^*$ on $\g_0^*$ is  given by
\begin{eqnarray}\label{action10}
\langle \xi\triangleright \alpha, x\rangle:=\langle \xi\wedge \phi^*\alpha,x\triangleright \mu\rangle,\qquad \xi\in \g_1^*,\alpha\in \g_0^*,x\in \g_0.
\end{eqnarray}
 If $[\mu,\mu]=0$, the coboundary Lie 2-bialgebra is called a {\bf triangular Lie 2-bialgebra}. 
 If the Lie 2-algebra $\g$ is integrated to a Lie 2-group $\mathbb{G}$, we further obtain the triangular Poisson 2-group $(\mathbb{G},\overleftarrow{\lambda_\mu}^{gpd}-\overrightarrow{\lambda_\mu}^{gpd})$, where $\lambda_\mu:G_0\to \wedge^2 \g_1$ is a $2$-section of the Lie algebroid $G_0\ltimes \g_1\to G_0$ defined by
\[\lambda_\mu(g_0)=g_0^{-1}\triangleright \mu-\mu,\qquad \forall g_0\in G_0.\]
See \cite[Theorem 3.11]{CSX2} for more details.

\begin{Cor}
Let $(\mathbb{G},\overleftarrow{\lambda_\mu}^{gpd}-\overrightarrow{\lambda_\mu}^{gpd})$ be a triangular Poisson 2-group. 
Then a maximal isotropic subspace $L$ with characteristic pair $(\h_1,r)$ for $r\in \wedge^2 \g_1$ defines a Dirac structure of $(\g_1,\g^*_1)$ if and only if 
the following conditions holds
\begin{itemize}
\item[\rm (1)] $\h_1\subset \g_1$ is a Lie subalgebra;
\item[\rm (2)]  $[\mu+r,\mu+r]\equiv 0$ ($\mathrm{mod}$ $\h_1$);
 \item[\rm (3)] 
 $[u,\mu+r]\equiv 0$ ($\mathrm{mod}$ $\h_1$) for all $u\in \h_1$. 
\end{itemize}
In this case, $H_0\subset \mathbb{G}$ is a Poisson subgroup if and only if $(\phi\wedge 1)(x\triangleright \mu)\in \h_0\wedge \g_{-1}$  for all $x\in \h_0$. On this condition, Poisson homogenous spaces $\mathbb{G}/\mathbb{H}$ correspond to $H_0$-invariant Dirac structures of $(\g_1,\g_1^*)$.
\end{Cor}
\begin{proof}
The first part of the corollary follows directly from  \cite[Corollary 3.2]{L}. Then $H_0\subset \mathbb{G}$ is a Poisson subgroup if and only if $\h_0^\perp\subset \g_0^*$ is a $\g_{-1}^*$-submodule. By \eqref{action10}, taking $\xi\in \g_{1}^*,\alpha\in \h_0^\perp$, we have $\xi\triangleright \alpha\in \h_0^\perp$ if and only if $\langle \xi\wedge \alpha, (1\wedge \phi)(x\triangleright  \mu)\rangle=0, \forall x\in \h_0$,
which amounts to say $(1\wedge \phi)(x\triangleright \mu)\in \h_0\wedge \g_{1}$  for all $x\in \h_0$. The rest of this result holds by Theorem \ref{Main}.
\end{proof}

\begin{Ex}
When $\mathbb{H}=\{e\}$ in Theorem \ref{Main}, the Poisson homogeneous space of the Poisson 2-group $(\mathbb{G}, \Pi_{\mathbb{G}})$ is  $\mathbb{G}$ itself with the affine Poisson structure  $\Pi_{\mathbb{G}}+\overleftarrow{r}$ (\cite{LLS, W}), where  $r\in \wedge^2 \g_1$ satisfies that $d_{\g_1^*} r+\frac{1}{2}[r,r]=0$. The corresponding Dirac structure of $(\g_1,\g_1^*)$ is $\mathrm{Graph}(r^\sharp)$.
\end{Ex}
\begin{Ex}
Let $(\mathbb{G},\Pi_{\mathbb{G}})$ be a Poisson 2-group with a closed 2-subgroup $\mathbb{H}$ satisfying that $H_0\subset \mathbb{G}$ is a Poisson subgroup. Then $(\mathbb{G}/\mathbb{H},\mathrm{pr}_*\Pi_{\mathbb{G}})$ is a Poisson homogeneous space of the Poisson 2-group $\mathbb{G}$ if and only if $\h_1^\perp \subset \g_1^*$ is a Lie subalgebra. The corresponding Dirac structure of $(\g_1,\g_1^*)$ is $\h_1\oplus \h_1^\perp$. In this case, $(\mathbb{G}/\mathbb{H},\mathrm{pr}_*\Pi_{\mathbb{G}})$ is a quotient Poisson groupoid.
\end{Ex}
\begin{Ex}
Let $\mathbb{G}$ be a Poisson 2-group with trivial Poisson structure and  assume that  $\mathbb{H}\subset \mathbb{G}$ is a closed 2-subgroup. Then there is a one-one correspondence between Poisson homogeneous spaces $(\mathbb{G}/\mathbb{H},\mathrm{pr}_*\overleftarrow{r})$ of $\mathbb{G}$ and $H_0$-invariant Dirac structures $L$ of $(\g_1,\g^*_1)$ with characteristic pairs  $(\h_1,r)$, in which $r\in \wedge^2 \g_1$ satisfies that 
\[[r,r]\equiv 0 \pmod{\h_1}, \qquad [\h_1^\perp, \h_1^\perp]_r\subset \h_1^\perp,\qquad H_0\triangleright r\equiv r \pmod{\h_1}.\]
\end{Ex}

\begin{Ex}
Let $(\g,\g^*)$ be a Lie 2-bialgebra with trivial Lie 2-algebra structure on $\g$. Then the abelian 2-group $\mathbb{G}=\g$ with the linear Poisson structure $\Pi_{KKS}$ is a Poisson 2-group. Let $\h:\h_1\to \h_0$ be a 2-subspace of $\g$ such that $\g_1^*\triangleright \h_0^\perp\subset \h_0^\perp$. Then there is a one-one correspondence between Poisson homogeneous spaces $(\mathrm{\g}/\mathrm{\h}, \mathrm{pr}_*(\Pi_{KKS}+r))$ and  Dirac structures of $(\g_1,\g_1^*)$ with characteristic pair $(\h_1,r)$. Here $r$ is an element in $\wedge^2\g_1$ satisfies that $d_{\g_1^*} r\equiv 0 \pmod{\h_1}$ and $[\h_1^\perp,\h_1^\perp]\subset \h_1^\perp$. 
\end{Ex}


\end{document}